\documentclass[a4paper,reqno]{amsart}
\pdfoutput=1

\usepackage{amsmath,amsfonts,amssymb,amsthm,stmaryrd,bbm,graphicx,mathtools,enumerate}
\usepackage{mathrsfs}
\usepackage{enumitem}
\usepackage[T1]{fontenc} 
\usepackage[latin1]{inputenc}
\usepackage[english]{babel}
\usepackage[tracking,spacing,kerning,babel]{microtype}
\usepackage{wasysym} 
\usepackage{color}
\usepackage{xcolor}
    	\usepackage{framed} 
\usepackage{wrapfig} 

\usepackage[final]{hyperref}   
\hypersetup{
    linktoc=page,
    linkcolor=red,          
    citecolor=blue,        
    filecolor=blue,      
    urlcolor=cyan,
   colorlinks=true           
}

\usepackage{lipsum} 

\usepackage{minitoc}
\usepackage{multicol}

\newtheoremstyle{mine}
{\baselineskip}
{\baselineskip}
{\itshape}
{
}
{\bfseries}
{.}
{.5em}
{#1 #2\ifx#3\relax\else~(#3)\fi}

\newtheorem{theorem}{Theorem}

\newtheorem{proposition}[theorem]{Proposition}
\newtheorem{lemma}[theorem]{Lemma}
\newtheorem{definition}[theorem]{Definition}
 
\theoremstyle{remark}
\newtheorem{remark}{Remark}

\colorlet{shadecolor}{blue!10}

\def\rm{\reversemarginpar}

\let\qed=\QED

\renewcommand{\epsilon}{\varepsilon}

\newcommand{\R}{\mathbb{R}}

\newcommand{\C}{\mathbb{C}}
\newcommand{\Z}{\mathbb{Z}}
\newcommand{\N}{\mathbb{N}}

\def\H{\mathbb{H}}

\def\calB{\mathcal{B}}

\def\calF{\mathcal{F}}

\def\calL{\mathcal{L}}

\def\calR{\mathcal{R}}

\def\dist{\mathrm{dist}}

\def\P{\mathbb{P}} 
\def\E{\mathbb{E}} 
\def\md{\mid}

\def \eps {\epsilon}

\def\Bb#1#2{{\def\md{\bigm| }#1\bigl[#2\bigr]}}

\def\Eb{\Bb\E}

\def\FK#1#2#3{{\def\md{\bigm| } \P_{#1}^{\,#2}  \bigl[  #3 \bigr]}}

\def \p {{\partial}}

\def\<#1{\langle #1\rangle}

\newcommand{\n}{{\mathbf n}}

\def\bi{\begin{itemize}}  
\def\ei{\end{itemize}}
\def\bnum{\begin{enumerate}} 
\def\enum{\end{enumerate}}
\def\ni{\noindent}

\def\Book{\mathbb{B}}

\title{Long-range order for critical Book-Ising and Book-percolation}

\author{Hugo Duminil-Copin, Christophe Garban, Vincent Tassion}

\address
{Universit\'e de Gen\`eve, 2-4 rue du Li\`evre, 1204 Gen\`eve, Switzerland, Institut des Hautes \'Etudes Scientifiques, 35 route de Chartres, 91440 Bures-sur-Yvette, France}
\email{hugo.duminil@unige.ch,duminil@ihes.fr}

\address
{Université Claude Bernard Lyon 1, CNRS UMR 5208, Institut Camille Jordan, 69622 Villeurbanne, France \, and Institut Universitaire de France (IUF)}
\email{garban@math.univ-lyon1.fr}

\address
{ETH Zurich, Department of Mathematics, Group 3
HG G 66.5
Rämistrasse 101,
8092 Zurich,
Switzerland}
\email{Vincent.Tassion@math.ethz.ch}

\begin{document}

\maketitle

\begin{abstract}
In this paper, we investigate the behaviour of statistical physics models on a book with pages that are isomorphic to half-planes. We show that even for models undergoing a continuous phase transition on $\mathbb Z^2$, the phase transition becomes discontinuous as soon as the number of pages is sufficiently large. In particular, we prove that the Ising model on a three pages book has a discontinuous phase transition (if one allows oneself to consider large coupling constants along the line on which pages are glued).
   Our work confirms predictions in theoretical physics which relied on renormalization group, conformal field theory and numerics (\cite{Car91,ITB91,SMP10}) some of which were motivated by the analysis of the Renyi entropy of certain quantum spin systems. 
\end{abstract}

\section{Introduction}
Consider the $N$-{\em pages book} $\Book_N$ obtained by gluing  $N$ copies of an upper-half plane $\mathbb H:=\Z\times \N$ along the bottom line $\Z \times \{0\}$, which is identified with $\mathbb Z$, see Figures \ref{f.branque} and \ref{f.BookIsing}. We call these copies the {\em pages} $\mathbb H^1,\dots,\mathbb H^N$ of the book and identify $\mathbb H^1$ with $\mathbb H$. 

Our goal is to explore the behaviour of classical statistical physics systems on a $N$-pages book. Of prime interest to us will be the family of Potts models as well as their corresponding graphical representations named Fortuin-Kasteleyn percolations.

\begin{figure}[!htp]
\begin{center}
\includegraphics[width=\textwidth]{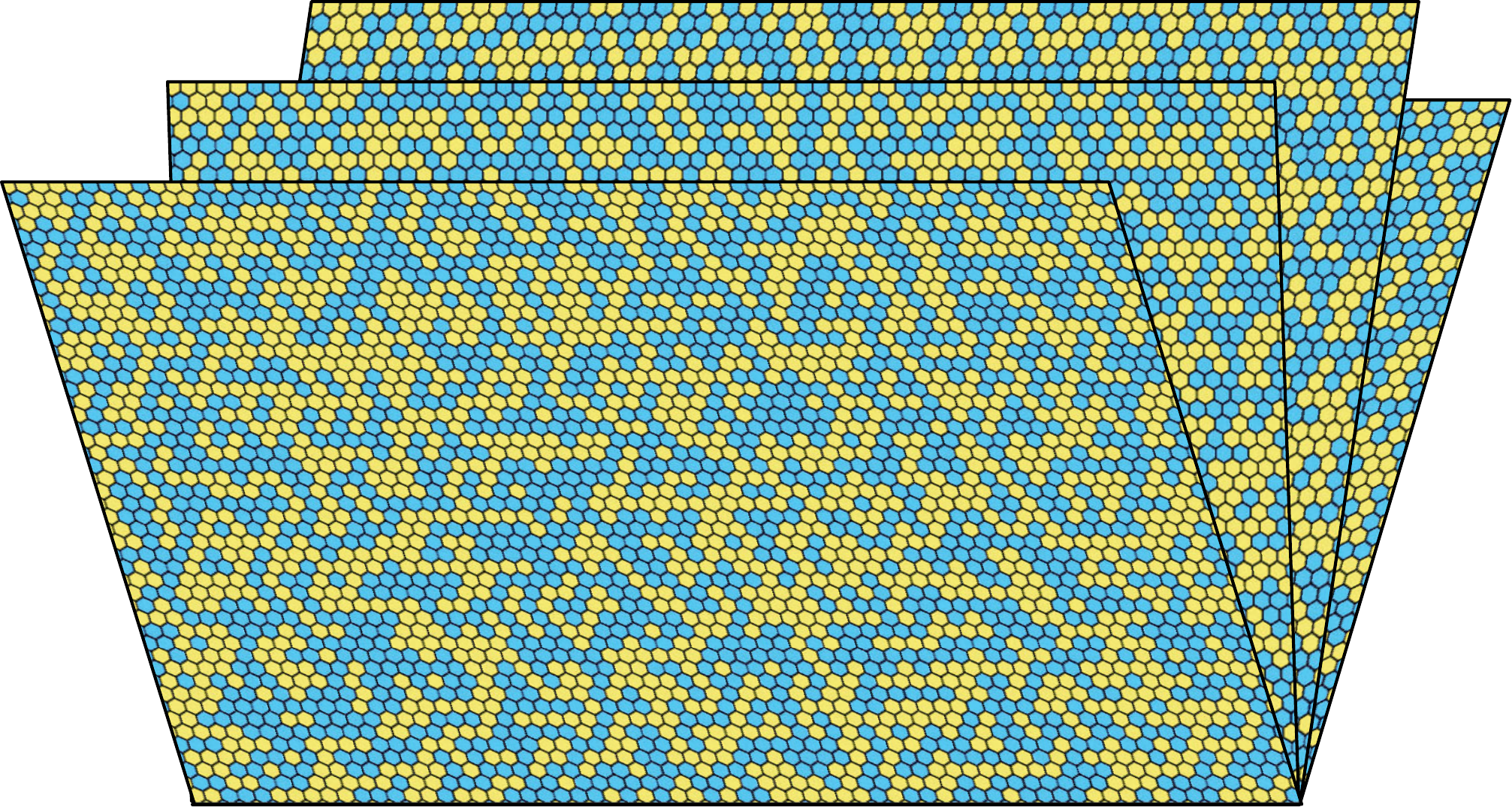}
\end{center}
\caption{Critical site percolation on the book $\Book_4$ (or rather its triangular lattice version here). The precise way of gluing the pages together does not impact our results.}\label{f.branque}
\end{figure}

\subsection{Potts model on the book.}

The Potts models are archetypical examples of statistical physics systems undergoing a phase transition in two dimensions. Fix an integer $q\ge2$. For $G=(V,E)$ a finite graph of an infinite graph $\mathbb G=(\mathbb V,\mathbb E)$ (we sometimes write $x\sim y$ if $xy\in E$), attribute a {\em spin} variable $\sigma_x$ belonging to a certain set $\Sigma:=\{1,2,\dots,q\}$ to each vertex $x\in V$. When $q=2$, one speaks of the Ising model and use $\{-,+\}$ instead of $\{1,2\}$ for $\Sigma$.
 A
 {\em spin configuration} $\sigma=(\sigma_x:x\in V)\in\Sigma^{V}$ is given by the collection of all the spins.
 Introduce the Hamiltonian of $\sigma$ for {\em free boundary conditions} defined by
 \begin{equation}\label{eq:H}
H_G^f(\sigma):=-\sum_{xy\in E}\,\mathbbm1[\sigma_x=\sigma_y]
\end{equation}
 corresponding to a ferromagnetic nearest-neighbor interaction. 
  For $\tau\in \Sigma$, we also define the Hamiltonian for {\em monochromatic $\tau$ boundary conditions}:
  \begin{equation}\label{eq:H}
H_G^\tau(\sigma):=H_G^f(\sigma)-\sum_{x\in V,y\in \mathbb V\setminus V:
x\sim y}\,\mathbbm1[\sigma_x=\tau].
\end{equation}
The above Hamiltonian corresponds to a ferromagnetic nearest-neighbor interaction. 
The {\em Gibbs measure on $G$ at inverse temperature $\beta\ge0$ with $\#$ {\em (where $\#$ is either free or monochromatic free)} boundary conditions} is defined by the formula
\begin{equation}\label{eq:Gibbs}\mu_{G,\beta}^\#[f]:=\frac{\displaystyle\sum_{\sigma\in\Sigma^{V}}f(\sigma)\exp[-\beta H_{G}^\#(\sigma)]}{\displaystyle\sum_{\sigma\in\Sigma^{V}}\exp[-\beta H_{G}^\#(\sigma)]}\end{equation}
for every $f:\Sigma^{V}\rightarrow \mathbb R$.

When $\mathbb G=\mathbb Z^2$ or $\mathbb B_N$, one may define the {\em Gibbs measure on $\mathbb G$ at inverse-temperature $\beta\ge0$ with  $\#$ boundary conditions} by taking the limit as $G\nearrow \mathbb G$ of the previous measures. In infinite volume, the model undergoes a phase transition on $\mathbb Z^2$ and $\mathbb B_N$ at some common $\beta_c=\beta_c(q)=\tfrac12\log(1+\sqrt q)$ \cite{BefDum12} in the following sense. If 
\begin{equation}
m_{\mathbb G}(\beta,q):=\mu_{\mathbb G,\beta}^\tau[\sigma_0=\tau]-\tfrac1q\end{equation}
 is the spontaneous magnetization of the model, then $m_{\mathbb G}(\beta,q)$ is equal to 0 if $\beta<\beta_c$ and is strictly positive if $\beta>\beta_c$.

When $\mathbb G=\mathbb Z^2(=\mathbb B_2)$, whether the phase transition is {\em continuous} (i.e.~$m_\mathbb G(\beta_c,q)=0$) or {\em discontinuous} (i.e.~$m_\mathbb G(\beta_c,q)>0$) has been the object of much interest in the past fifty years. It was proved in \cite{Ons44,Yan52} that the phase transition of the Ising model is continuous on $\mathbb Z^2$. More generally, it was predicted by Baxter \cite{Bax73} that the phase transition of the Potts model on $\mathbb Z^2$ is continuous for $q\in\{2,3,4\}$ and discontinuous for $q>4$. See \cite{DumGanHar16,DumSidTas13} for a proof of this statement (see also \cite{RaySpi19} for a short proof in the case $q>4$).

In this paper, we investigate the question on $\mathbb B_N$ and prove the following result.

\begin{theorem}\label{thm:Potts}
There exists $N_0<\infty$, such that for every $q\in\{2,3,4\}$, the $q$-Potts model undergoes a first-order phase transition on $\Book_{N_0}$. Equivalently, for every $N\geq N_0$ and every $q\in\{2,3, 4\}$, 
\[
m_{\mathbb B_{N}}(\beta_c,q)>0.
\]
\end{theorem}

As we shall explain below, it is natural in several respects to allow ourselves to strengthen the coupling constants along the edges of the gluing line $\Z$.

 For $J\ge 0$ and $G\subset \mathbb G$, we then introduce the modified measure $\mu_{G,\beta,J}$ where $H_G^\tau$ is replaced by the Hamiltonian

\begin{align*}\label{}
H_{G,J}^\tau(\sigma)&=H_G^\tau(\sigma)-(J-1)\sum_{x\sim y\in V\cap \mathbb Z}\mathbbm 1[\sigma_x=\sigma_y]
\end{align*}
(corresponding to changing coupling constants along the line $\mathbb Z$ from 1 to $J$) and the associated quantities $\mu_{G,\beta,J}^\tau$, $\mu_{\mathbb G,\beta,J}^\tau$ and $m_\mathbb G(\beta,J,q)$.

To motivate the introduction of the parameter $J$, let us briefly mention the slightly related problem of long-range Potts model on $\Z$. The previous procedure is the analog of strengthening the coupling-constants between adjacent vertices in this context: As an example, in \cite{aizenman1988}, coupling-constants are defined as $J_{x,y}=J_{x,y}(J):=J\, 1_{x\sim y}+\frac 1 {|x-y|^2}1_{|x-y|\geq 2}$ and the following critical point is introduced (\cite{aizenman1988,ImbrieNewman}, see also our recent work \cite{DGT20a}), 
\begin{align*}\label{}
\beta^*(q):=\inf\{ \beta \text{ s.t. } \exists J <\infty \, \text{ for which there is long-range order for $\{J_{x,y}(J))\}_{x,y}$}\}.
\end{align*}

In our present context, motivated by the predictions from \cite{Car91,ITB91,SMP10,Ste14} (see Subsection \ref{ss.motiv} below), and by analogy with $\beta^*(q)$, we  define below a notion of ``optimal'' number of pages $N^*(q)$ needed to create a first-order phase transition. The advantage of the notions $\beta^*(q)$ and $N^*(q)$ comes from the fact that they are universal: they do not depend on the particular way of gluing pages together (as far as the glue is finite-range, say) or even the underlying lattice (it could be triangular or hexagonal for instance).  
For any $q\in[1,4]$, define 
\begin{align}\label{e.Nq}
N^*(q):= \min\{ N\in \N, \; \exists J<\infty\,\, \text{so that  } 
m_{\Book_{N}}(\beta_c,J,q)>0\} \,.
\end{align}

We obtain the following result on the behavior of the optimal number of pages $N^*(q)$ depending on $q$.

\begin{theorem}\label{thm:Potts*} $ $ We have the following:
\bi
\item[(i)] $N^*(2)=3$ 
\item[(ii)] $N^*(3)=2$ 
\item[(iii)] $1\leq N^*(q) \leq 2$ for all $q\geq 4$. 
\ei
\end{theorem}
We will discuss each of these items below, after Theorem \ref{thm:FK2} which is the analogous statement for FK percolation with cluster-weight $q\in [1,\infty)$.

\begin{remark}\label{r.phys}
As we will explain further below, the case $q=2$ turns out to be especially interesting. Physicists which considered this question have predicted that the first-order transition in fact arises as soon as the number of pages is ``$2+\eps$''. See Remark \ref{r.2eps} and Subsection \ref{ss.motiv}.
\end{remark}

\ni
Interestingly, in the case $q=2$, the effect of this first-order phase transition is to make the Ising model on each of the $N$-pages independent of each other in the scaling limit. The statement below (written for $N=3$, but it would also be valid for large enough $N$ and $J=1$) makes this factorization property more precise. Below, for a set $A\subset V$, write $\sigma_A:=\prod_{x\in A}\sigma_x$.

\begin{theorem}\label{thm:Ising2}
Fix $q=2$. Let $J<\infty$ be large enough  so that $
m_{\mathbb B_3}(\beta_c,J,2)>0$. For any three sets $A\subset \mathbb H^1$, $B\subset\mathbb H^2$, and $C\subset\mathbb H^3$, containing a total of $m$ vertices that are all at a distance at least  $L$ of  $\mathbb Z$, we have the following factorization property of $m$-point correlations across $\mathbb Z$: 
\begin{align*}
\mu_{\mathbb B_3,\beta_c,J}^f[ \sigma_A \sigma_B \sigma_C ]&= \mathbf 1_{m\in2\Z}\; \mu_{\H,\beta_c}^{+}[\sigma_A]\mu_{\H,\beta_c}^{+}[\sigma_B]\mu_{\H,\beta_c}^{+}[\sigma_C](1+O_m((\log L)^{-c}))\,.
\end{align*}
If $+$ boundary conditions are prescribed instead, the condition on the parity of $m$ can be dropped and we get
\begin{align*}
\mu_{\mathbb B_3,\beta_c,J}^{+}[ \sigma_A \sigma_B \sigma_C ]&=  \mu_{\H,\beta_c}^{+}[\sigma_A]\mu_{\H,\beta_c}^{+}[\sigma_B]\mu_{\H,\beta_c}^{+}[\sigma_C](1+O_m((\log L)^{-c}))\,.
\end{align*}
\end{theorem}

Let us mention that the error term $(\log L)^{-c}$ can be improved by looking more closely at our proof, but this is irrelevant for the conclusion of the paper.
\subsection{Fortuin-Kasteleyn percolation on the book.}

We now define the Fortuin-Kasteleyn percolation \cite{ForKas72,For70} (we also refer to \cite{Gri06} for a manuscript and \cite{Dum17a} for recent results). Let  $G=(V,E)$ be a subgraph of an infinite graph $\mathbb G$, let $\xi$ be a partition of the vertices $\partial G:=\{x\in V:\exists y\in \mathbb G\setminus V:xy\in \mathbb E\}$. A {\em percolation configuration} is an element $\omega=(\omega_e:e\in E)\in \{0,1\}^E$. If $\omega_e=1$ we say that the edge is {\em open}, otherwise it is {\em closed}. We often see $\omega$ as a subgraph of $G$ with vertex-set $V$ and edge-set given by the set of open edges in $\omega$.

The FK percolation measure on G with edge-weights $(p,\lambda)$ and cluster-weight $\xi$ is defined by the formula
\[
\mathbb P_{G,p,\lambda,q}^\xi[\omega]=\frac{q^{k(\omega^\xi)}}Z \prod_{xy\in E}p_e^{\omega_e}(1-p_e)^{1-\omega_e},
\] 
with $p_e=p$ if at least one endpoint is not in $\mathbb Z$, and $\lambda$ if both are, and where  $\omega^\xi$ is the graph obtained from $\omega$ by wiring all the vertices in $\partial G$ belonging to the same element of the partition $\xi$. Let $\xi=1$ (resp.~$\xi=0$) be  the wired (resp.~free) boundary conditions corresponding to the partitions equal to $\{\partial G\}$ (resp.~only singletons). 

Below, we will use the notation $A\longleftrightarrow B$ (in $C$) if there exists a path of open edges between a vertex in $A$ and a vertex in $B$ (using vertices in $C$ only). We also write $x$ instead of $\{x\}$ when the set is a singleton, and $x\longleftrightarrow\infty$ to denote the fact that there exists an infinite path starting from $x$. 

We construct the FK percolation $\mathbb P_{\mathbb G,p,\lambda,q}^1$ and $\mathbb P_{\mathbb G,p,\lambda,q}^0$ on $\mathbb G$ with wired or free boundary conditions by taking the limit as $G\nearrow \mathbb G$ of the measures $\mathbb P_{G,p,\lambda,q}^1$ and $\mathbb P_{G,p,\lambda,q}^0$. We also define, for an infinite graph $\mathbb G$ containing the origin, 
\[
\theta_{\mathbb G}(p,\lambda,q):=\mathbb P_{\mathbb G,p,\lambda,q}^1[0\longleftrightarrow \infty].
\]
It was also proved that there exists $p_c=p_c(q)=\sqrt q/(1+\sqrt q)$ such that for every integer $N$ and $\lambda\in(0,1)$, $\theta_{\mathbb B_N}(p,\lambda,q)$ is equal to 0 if $p<p_c$ and is strictly positive if $p>p_c$. Again, the question of whether the phase transition is continuous (i.e.~$\theta_{\mathbb B_N}(p_c,\lambda,q)=0$) or discontinuous (i.e.~$\theta_{\mathbb B_N}(p_c,\lambda,q)>0$) was answered in the special case of $\mathbb G=\mathbb Z^2(=\mathbb B_2)$: when $1\le q\le 4$, it is continuous \cite{DumSidTas13} and when $q>4$, it is discontinuous \cite{DumGanHar16}. Here, we investigate this question on $\mathbb B_N$ with $N\ge 3$. Our first result is as follows.

\begin{theorem}\label{thm:FK}
For any $1\leq q \leq 4$, there exists $N_0<\infty$ such that FK percolation undergoes a first-order phase transition on $\Book_{N_0}$. I.e. for any $N\geq N_0$,  
\[
\theta_{\Book_{N}}(p_c(q),q)=\mathbb P_{\Book_N,p_c(q),q}^1[0\longleftrightarrow \infty] >0\,.
\]
By choosing $N_0$ sufficiently large, the result also holds for arbitrary small $\lambda \in [0,1)$ and for free boundary conditions, i.e.  for any $N\geq N_0$, 
\[
\mathbb P_{\Book_N,p_c(q),\lambda,q}^{0}[0\longleftrightarrow \infty] >0\,.
\]
\end{theorem}


Note that Theorem~\ref{thm:Potts} follows easily from Theorem~\ref{thm:FK}.
\begin{proof}[Proof of Theorem~\ref{thm:Potts}]Through the Edwards-Sokal coupling between the Potts model and FK percolation (see e.g.~\cite{Gri06}), we  have that 
\[
m_{\mathbb B_N}(\beta_c,q,J)=\tfrac{q-1}{q}\theta_{\mathbb B_N}(p_c,q,1-e^{-2\beta J}),
\]
hence Theorem~\ref{thm:Potts} is a direct consequence of Theorem~\ref{thm:FK}.
\end{proof}

As in the case of Potts models, we define for any $q\geq 1$, 
\begin{align*}\label{}
N^*(q):= \min\{ N\in \N, \; \exists \lambda<1\,\, \text{so that  }  \theta_{\Book_N}(p_c(q),\lambda,q)>0\} \,.
\end{align*}

The following result gives a precise picture of the optimal number of pages $N^*(q)$ depending on $q\geq 1$.  (See Fig.~\ref{f.Nq} for a plot of $q\mapsto N^*(q)$). This extends Theorem \ref{thm:Potts*} which was stated for Potts models ($q\in \N_+$).
\begin{figure}[!htp]
\begin{center}
\includegraphics[width=0.8\textwidth]{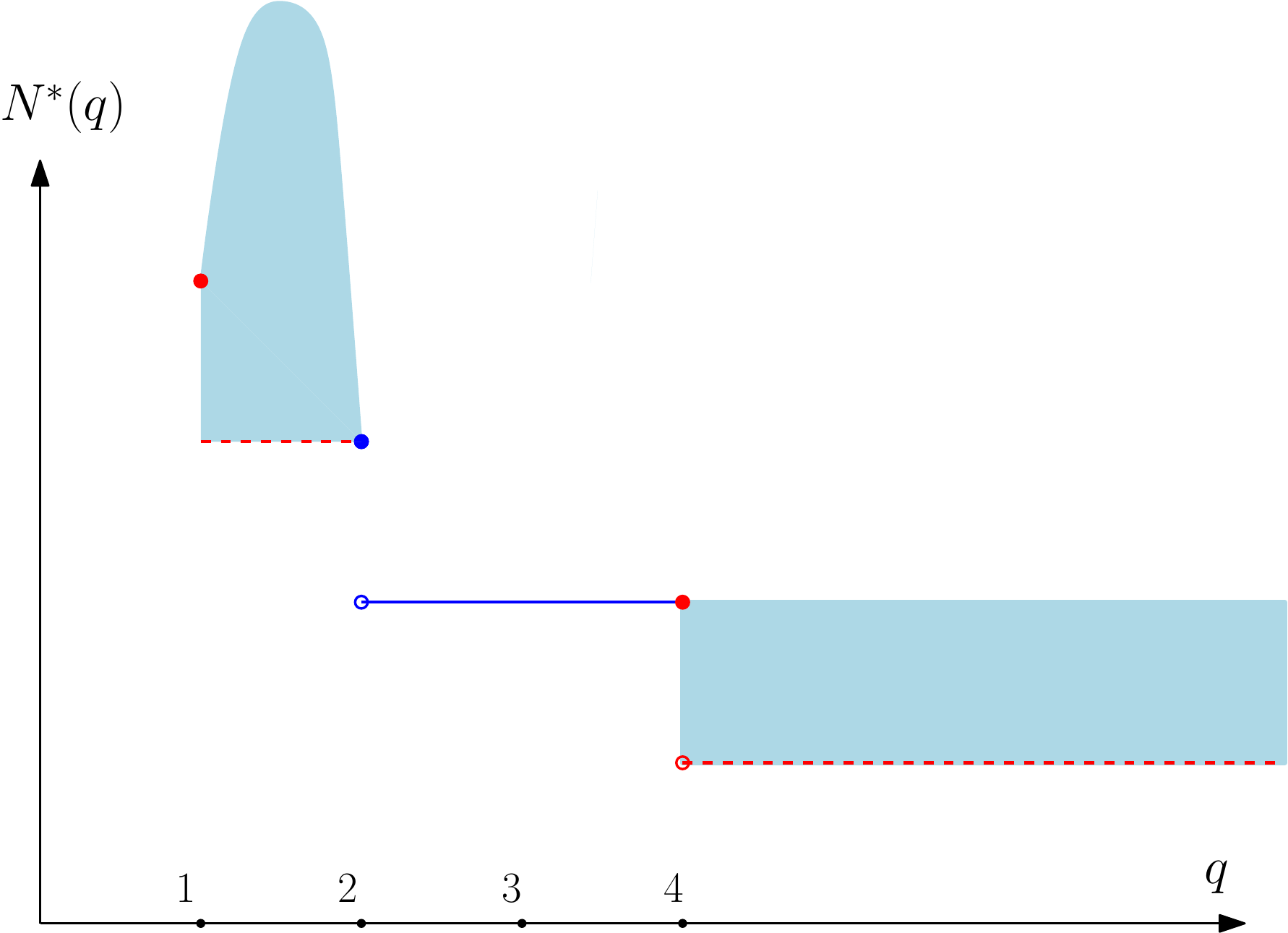}
\end{center}
\caption{The lines and dots in dark blue show the precise values proved for $N^*(q)$. The light blue shows the possible range of values for $N^*(q)$ while the red color indicates where we expect $N^*(q)$ to be. }\label{f.Nq}
\end{figure}

\begin{theorem}\label{thm:FK*} $ $ We have that
\bi
\item[a)] $3 \leq N^*(1) \leq 4$,
\item[b)] there exists $N_0$ such that $3 \leq N^*(q) \leq N_0$ for every $1\le q<2$,
\item[c)] $N^*(2)=3$,
\item[d)] $N^*(q)=2$ for all $2< q < 4$, 
\item[e)] $1\leq N^*(q) \leq 2$ for all $q\geq 4$. 
\ei
\end{theorem}


\ni
We now comment on the different items in the above result.
\bnum
\item[a)] The fact that $N^*(1)\leq 4$ will follow readily from our proof of Theorem \ref{thm:FK} using the known value of the {\em one-arm} critical exponent in $\H$ for critical $q=1$ percolation (\cite{SW01,PonIkh}). 
We expect that this is optimal, i.e.~that $N^*(q=1)=4$. The bound $N^*(1) \geq 3$ will be shown using a second moment argument in Section \ref{s.2dM}.

\item[b)] We provide a direct proof in Subsection \ref{ss.unif} that $\sup_{1\leq q \leq 4} N^*(q) <\infty$. The fact that $N^*(q)\geq 3$ when $1<q<2$ will also be proved in Section \ref{s.2dM} using a  second moment argument based on estimates on the one-arm critical exponents from \cite{DMT20}. We expect that $N^*(q) =3$ in this whole regime.

\item[c)] The case $q=2$ is, arguably, the most interesting of all. 
As opposed to the $q=1$ case, this result will not be a straightforward consequence of (the proof of) Theorem \ref{thm:FK}. Its proof will be organized as follows:
\bi
\item[1)] The proof that $N^*(2)\leq 3$ will be the focus of Section \ref{s.Ising}. The argument will be based on the {\em random currents} representation of the Ising model (\cite{Aiz82,Dum16}). Random currents will indeed enable us to show that in the graph $\Book_3$, far from the middle line $\Z$, the spin system behaves (nearly) as if all edges along $\Z$ were wired together. This will be a key step of the proof as the precise values of {\em arm-exponents} in $\H$ depend on what are the boundary conditions induced along $\p \H$.
\item[2)] The second part of the proof is to show that $N=2$ pages are not sufficient to create an infinite cluster even if the edge-weights $\lambda$ are arbitrary large on $\Z$. Here, the second moment argument used for the case $1\leq q<2$ is not sufficient and a detailed analysis of the effect of a $1d$ defect-line for $2d$ critical Ising model is needed. This will be the subject of the companion paper \cite{DGT20b}.
\ei
\item[d)] The proof that $N^*(q)=2$ for any $2<q<4$ will consist in showing that a defect line $\Z$ with high coupling constants $\lambda$ is sufficient in $\Book_2=\Z^2$ to create on its own an infinite cluster. The proof is given in Section \ref{s.2dM}.  It will rely on the multiscale/renormalization argument built for Theorem \ref{thm:FK} but will be simpler due to the planarity of $\Book_2$.  
\item[e)] Finally, when $q>4$, it follows from the first-order phase transition (\cite{DumGanHar16,RaySpi19}) that $N^*(q)\leq 2$ for all $q>4$ (no strengthening $\lambda$ along $\Z$ is needed in that case) and we expect that $N^*(q)=1$ in this regime.
When $q=4$, the  argument of item e) still works to ensure $N^*(4)\leq 2$ but strong RSW is missing to check that $N^*(4)\geq 2$. We still expect though that $N^*(4)=2$. 
\enum

\begin{remark}\label{r.2eps}
When $q=1$ (resp.~$q=2$), it is not difficult to extend the analysis carried in this paper to a Book with $``N=3+\eps"$ pages (resp $``N=2+\eps"$ pages) in the following sense: consider the finite book with 3 pages (resp. 2 pages) of normal size $[0,n]\times [0,n]$ and a fourth (resp third) page of size $[0,n]\times [0,n^\eps]$. These  pages are glued along $[0,n]$. We claim that by a slight adaptation of the multiscale proof in this paper, we can show that if the coupling constant $\lambda$ is chosen high enough along $[0,n]$, then with probability $1-o(1)$ as $n\to \infty$, there is a {\em macroscopic} cluster in the $``N=3+\eps"$ book (resp long-range order in the $``N=2+\eps"$ book) with intensity larger than $\tfrac34$ along the gluing line $[0,n]$. This is consistent with predictions from \cite{SMP10,Ste14} (though with a different notion of $``N=2+\eps"$ pages). 
\end{remark}

{\em
In the whole paper, we focus on $1\leq q \leq 4$ 
and $p=p_c$. We drop them from the notation. In particular we write $\mathbb P^\xi_{G,\lambda}$ instead of $\mathbb P^\xi_{G,p_c,\lambda,q}$. It will happen that we write $\mathbb P^\xi_{G,p_c}$, but we warn the reader that this means that the $\lambda$ parameter is equal to $p_c$ (as the $p$ is always set to $p_c$).
}

\subsection{Motivations from replicas and quantum spin systems.}\label{ss.motiv}

Our results are motivated by several works in theoretical physics. To our knowledge, the first works which have considered the present gluing problem are the works \cite{Car91,ITB91} by Cardy and Iglói-Turban-Berche. These two works rely on a renormalization group analysis in order to study the large $N$ tends to infinity case. Based on this RG analysis, both \cite{Car91} and  \cite{ITB91} suggest that if one glues an Ising model at $\beta_c$ on $N>2$ pages along a line, then the spins may spontaneously order near that line. The gluing of several pages of Ising arises naturally  in their works in forms of {\em replicas} for a model with disorder, namely a $2d$ Ising model with {\em quenched} magnetic disorder along its boundary $\p \H$. 

\begin{figure}[!htp]
\begin{center}
\includegraphics[width=\textwidth]{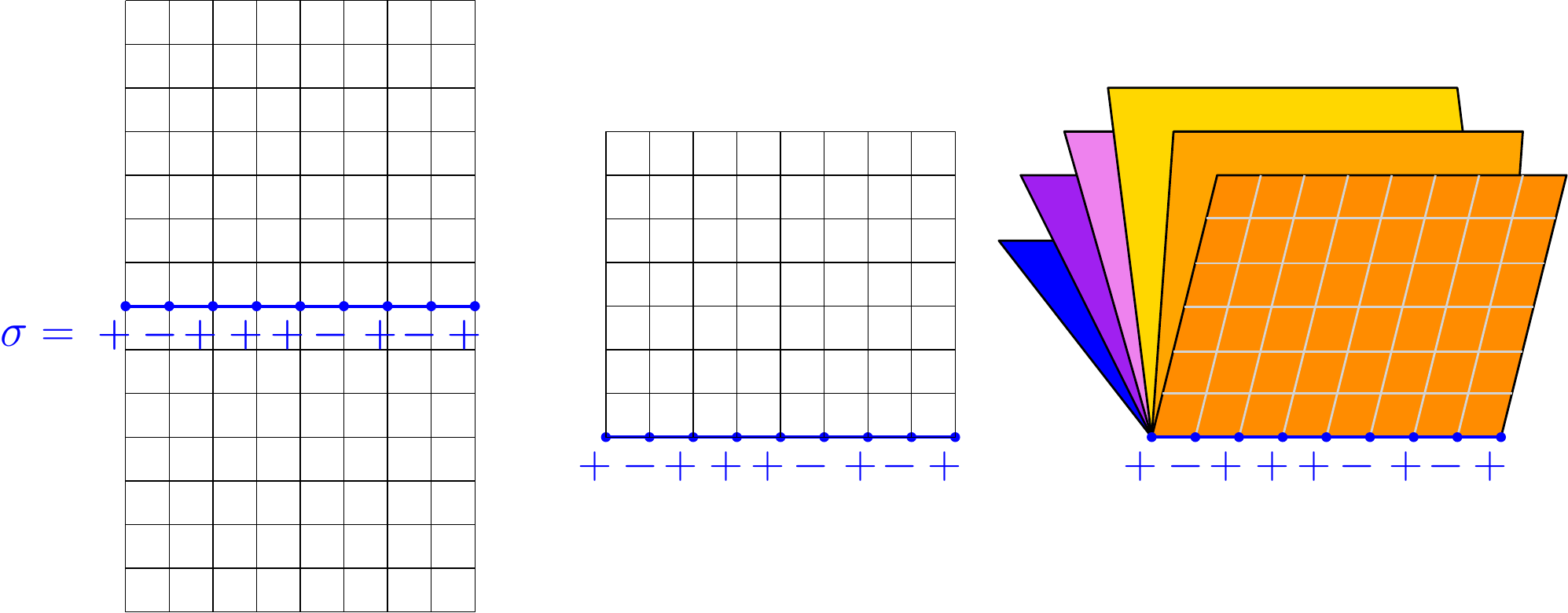}
\end{center}
\caption{If $p_\sigma$ is the probability to find the configuration $\sigma$ on the middle line on the left, then the probability to find the same $\sigma$ at the bottom of the page in the middle is proportional to $p_\sigma^{1/2}$ while the probability to find $\sigma$ at the bottom of the Book-graph on the right is proportional to $p_\sigma^3$.}\label{f.BookIsing}
\end{figure}

More recently, in the works \cite{SMP10, Ste14} by Stéphan-Misguich-Pasquier and Stéphan, the authors combine conformal field theory arguments with numerical computations in order to give strong further support to these predictions. See also the simulations in \cite{Gra17}.

The goal behind \cite{SMP10,Ste14} is in some sense also driven by the {\em replica-trick} but for a different underlying motivation than in \cite{Car91,ITB91}. In these papers, the authors are interested in the {\em Shannon entropy}
of the groundstate $|\psi\rangle$ of the quantum Ising chain (or quantum Ising chain in transverse field), which is given on $\Z_L:=\Z/L\Z$ by the following Hamiltonian
\[
H_{\text{Ising Chain}}= -\sum_{i\in \Z/L\Z} \sigma_i^x\sigma_{i+1}^x + h\, \sigma_i^z\,.
\]
For this Ising chain, the most natural basis, denoted $\{|\sigma\rangle\}_{\sigma\in \{-1,1\}^L}$, of $(\C^2)^{\otimes \Z/L\Z}$ is given by the eigenstates of $\sigma_i^x$ which correspond to the actual spins in the classical two-dimensional model. In this  basis, and for the critical parameter $h:=h_c=1$ in the quantum Hamiltonian $H_{\text{Ising Chain}}$, the ground state can be written as 
\[
|\psi\rangle = \sum_{\sigma \in \{\pm 1\}^L}  p_\sigma^{1/2} |\sigma\rangle\,,
\]
where  $p_\sigma$ denotes the probability for a classical Ising model\footnote{For this correspondance to hold, the classical Ising model should not be on a $\Z^2$-grid but rather on a $\Z\times \R$ lattice. We will not enter into these considerations here, but simply mention that our study does extend to this more general framework using a similar renormalization framework and the statement of \cite{duminil2018universality} guaranteeing that the behaviour on $\Z^2$ is similar as the one on $\Z\times\R$.} in the infinite $2d$ cylinder  $\Z_L\times \Z$ to generate at $\beta_c$  the configuration $\sigma$ at the middle slice of the cylinder $\Z_L\times \{0\}$. The {\em Shannon Entropy} of the Quantum Ising chain in the basis is then defined as 
\[
S = - \sum_{\sigma} p_\sigma \log p_\sigma\,.
\]
The connection with Book-Ising goes as follows: one can express the entropy $S$ as a limit as $n\to 1$ of the so-called {\em Renyi's entropies} $S_n$:
\[
S= \lim_{n\to 1} S_n = \lim_{n\to 1} \frac 1 {1-n} \log \Big( \sum_\sigma p_\sigma^n\Big)\,.
\]
Now, in the  spirit of the celebrated Parisi replica's trick, the idea in \cite{SMP10,Ste14} is to analyze $S$ via the analysis of the Renyi entropies  $\{S_n\}_{n\in \N^*}$.  The link with Book-Ising is that the measure on $\sigma\in \{-1,1\}^L$ which assigns a weight on each configuration $\sigma$ proportional to $p_\sigma^n$ can be realized as a Book-Ising on $N=2n$ pages (where pages here are semi-infinite cylinders $\Z_L\times \N$). See Fig.~\ref{f.BookIsing} (with squares instead of semi-infinite cylinders).

%
%
%
%
%

\smallskip
\ni
\subsection*{Organization of the paper.} In Section~\ref{sec:2}, we present the preliminaries of the paper and the important disconnection exponents. At the core of this section is the statement of Proposition~\ref{thm:FK2}. Section~\ref{sec:3} contains the proof of Proposition~\ref{thm:FK2}. Sections~\ref{sec:4} and \ref{sec:5} contain the proofs of Theorem~\ref{thm:FK*} for $q\ne 2$ and $q=2$ respectively.

\smallskip
\subsection*{Acknowledgements.} The second author wishes to thank Jean-Marie Stéphan for very inspiring discussions on the physics side of this problem. The first author is funded by the ERC CriBLaM, the Swiss FNS and the NCCR SwissMap. The research of the second author is supported by the ERC grant LiKo 676999.  
The third author was funded by the ERC grant 851565.

\section{Preliminaries and disconnection exponent on the Book}\label{sec:2}

\subsection{Disconnection exponent.}
In the rest of the paper, depending on the context, $\Lambda_K$ will be either the box of size $K$ in $\mathbb Z^2$. 
We will extensively rely through this paper on the following event. For any $1\leq k \leq K$, let 
$F(k,K)$ be the event that there exists a page $\H^u$ in which $\partial\Lambda_k$ is disconnected from $\partial\Lambda_K$ in $\H^u$ by a path in $\omega$.  Let us mention that the complementary event $F(k,K)^c$ can also be interpreted using the dual representation of the Fortuin-Kasteleyn percolation on the page, where $\omega^*$ is defined as follows. For each page $\H^u$, let $(\H^u)^*$ be the dual graph of $\H^u$, and set $\omega^*_{e^*}=1-\omega_e$, where $e^*$ is the unique dual edge that crossed $e$ in its center. Then, we speak of a dual-open path of dual-edges for a path in $(\H^u)^*$ which is open in $\omega^*$ (we write $A\stackrel{*}{\longleftrightarrow}B$ for the existence of a dual connection between the sets $A$ and $B$). Then, $F(k,K)^c$ corresponds to the event that in each page, there exists a dual-open path from $\partial\Lambda_k$ to $\partial\Lambda_K$,
see Fig.~\ref{f.eventF}.
\begin{figure}[!htp]
\begin{center}
\includegraphics[width=\textwidth]{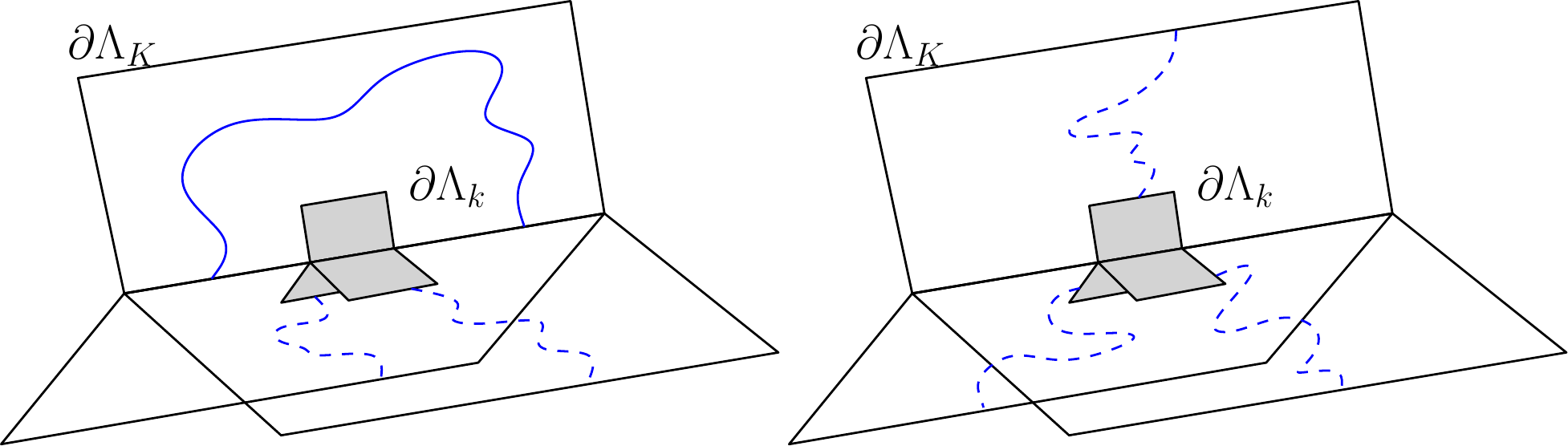}
\end{center}
\caption{The event $F(k,K)$ is realized on the left while $F(k,K)^c$ is realized on the right (the dashed lines correspond to dual open paths). When $F(k,K)^c$ will hold, it will disconnect the left side of the book from its right. 
}\label{f.eventF}
\end{figure}

Below, we will speak of a {\em critical exponent} $\alpha^*$ for a family of probabilities $(\mathbb P[A(k,K)]:k\le K)$ as follows 
\[
\alpha^*:=\sup\{\alpha>0:\exists \rho_0\text{ s.t.~$\forall K\geq 1, \rho\geq \rho_0$, }\mathbb P[A(K,\rho K)]\le \rho^{-\alpha}\}.
\]

Morally speaking, this critical exponent is ruling the speed of algebraic decay -- in $(k/K)$ -- of the probabilities $\mathbb P[A(k,K)]$. In what follows, we expect the families of probabilities (but this is currently unknown for most of the families under consideration) exhibit a behaviour of the form
\[
\mathbb P[A(k,K)]=(k/K)^{\alpha^*+o(1)},
\]
where $o(1)$ is a quantity that tends to 0 as $k/K$ tends to 0, but this is currently unknown for a number of them.

\begin{definition}\label{d.exp}
The \emph{disconnection exponent} $\alpha(q,N)$ is defined as the critical exponent of the family $\mathbb P_{B_{K},p_c,p_c,q}^0[F(k,K)^c]$.
\end{definition}

This disconnection exponent will be of central importance in this work as its value will exactly detect when (as $N$ increases) the phase-transition becomes first-order instead of second-order. Indeed the main ingredient for the proof of our main results.

\begin{proposition}\label{thm:FK2}
For every $1\le q\le 4$, if $N\ge1$ is such that  $\alpha(q,N)>1$, then there exists $\lambda\in(0,1)$ such that
\[
\theta_{\mathbb B_{N}}(p_c,\lambda,q)>0.
\]
In other words, $\alpha(q,N)>1$ implies $N^*(q) \leq N$. 
\end{proposition}

\subsection{Arm-exponents in $\H$.}\label{ss.arm}
 
The following three {\em one-arm exponents} in the upper-half plane will help us obtain estimates on the disconnection exponent $\alpha(N,q)$ uniformly in $1\leq q \leq 4$. As they are not known to exist, we define them like the disconnection exponent (in the notation below we ignore the parameter $\lambda$ as it is set to $p_c$):
\bi
\item $\alpha^+_{free}(q)$: the critical exponent for the family 
\[
a_{\rm free}^+(k,K,q):=\mathbb P_{\H,p_c,q}^0[\partial\Lambda_k\stackrel{*}{\longleftrightarrow}\partial\Lambda_K].
\]
\item $\alpha^+_{\C}(q)$ : the critical exponent for the family 
\[
a_{\C}^+(k,K,q):=\mathbb P_{\Z^2,p_c,q}^0[\partial\Lambda_k\stackrel{*}{\longleftrightarrow}\partial\Lambda_K\text{ in }\H].
\]
\item $\alpha^+_{wired}(q)$ : the critical exponent for the family 
\[
a_{\C}^+(k,K,q):=\mathbb P_{\H,p_c,q}^1[\partial\Lambda_k\stackrel{*}{\longleftrightarrow}\partial\Lambda_K].
\]
\ei
Note that with these definitions, the following special cases are known:
\bi
\item[i)] $\alpha^+_{free}(1)=\alpha^+_{\C}(1)=\alpha^+_{wired}(1)=\tfrac13$ as proved in \cite{SW01,PonIkh} respectively for triangular and $\Z^2$ lattices. 
\item[ii)] $\alpha^+_{wired}(2)=\frac 1 2$  (see e.g.~\cite{DHN11}). 
\ei
For future reference (we will use these estimates later on), we write $a_\#^+(K)$ instead of $a_\#^+(0,K)$.
\begin{remark}
Since the free (resp.~wired) boundary conditions are helping (resp.~disadvantaging) a dual connection, we have that $\alpha^+_{free}(q)\le \alpha^+_{\C}(q) \le \alpha^+_{wired}(q)$. We will note use this fact, but for all $1<q\leq 4$, one has 
$
\alpha^+_{free}(q)< \alpha^+_{\C}(q) < \alpha^+_{wired}(q)\,.
$
\end{remark}



\subsection{Proof of Theorem~\ref{thm:FK} given Proposition \ref{thm:FK2}.}\label{ss.unif}

We first prove the following uniform control on $N^*(q)$.
\begin{proposition}\label{pr.unif}
There exists $N_0$ such that for every $1\le q\le 4$, $N^*(q)\le N_0$.
\end{proposition}

\begin{proof}
Assuming Proposition \ref{thm:FK2} holds, it is enough to find an integer $N_0$ large enough so that
$
\alpha(q, N_0) >1$ for every $1\le q\le 4$.

The most trivial bound on $\alpha(q,N)$ is obtained as follows. For $F(k,K)$ not to occur, it must be that in each page, $\partial\Lambda_k$ is connected to $\partial\Lambda_K$ in the dual configuration $\omega^*$; see Fig.~\ref{f.eventF}. Using the comparison between boundary conditions, one may split the book into disconnected pages and use that this event has a probability smaller than $C\rho^{-\alpha^+_{free}(q)}$ in each page. This reasoning gives 
\[
\alpha(q,N)\ge N\alpha^+_{free}(q).\]
  It is known from \cite{PonIkh,SW01} that $\alpha_0(1)=\frac13$, so we already obtain at this stage a proof of the upper-bound in item a) of Theorem \ref{thm:FK*}, i.e 
\[
N^*(1)\le 4.
\]
For the remaining $1<q<4$, it was proved in \cite{DumSidTas13} that $\alpha^+_{free}(q)>0$, thus giving the existence of $N=N(q)$ such that $\alpha(q,N)>1$. The problem with this bound is that it deteriorates when $q$ tends to 4, for which $\alpha^+_{free}(4)$ is expected to be equal to 0.
This reasoning would force us to choose a number of pages $N(q)$ tending to $\infty$ as $q\nearrow 4$. 

A slightly better bound is obtained by observing that by successively conditioning in each page, for all but the last page, the probability of having a dual path connecting $\partial\Lambda_k$ is connected to $\partial\Lambda_K$ in a page is smaller than $a^+_{\C}(k,K)$ since there exist at least two undiscovered pages (and therefore by comparison between boundary conditions the occurrence of the connection is smaller than the one  under the full plane measure), which explains why we introduced above the exponent $\alpha^+_{\C}(q)$. (This observation will also be used in the proof of Lemma \ref{l.anchor}). This domination is valid as long as there are at least two remaining pages so we get
\[
\alpha(q,N)\ge (N-1)\alpha^+_{\C}(q).
\]
 This exponent is know from \cite{DumSidTas13} to be larger than some constant $c>0$ uniformly on $1\le q\le 4$. As a consequence, we deduce that $N^*(q)\le N_0$ uniformly in $1\le q\le 4$ which thus proves the content of Proposition \ref{pr.unif}.
\end{proof}

\begin{remark}\label{}
As we will obtain $N^*(q) \leq 2$ by other means when $q>2$ (in Section \ref{s.2dM}), we may have focused here only on the case $1<q<2$ which is slightly  simpler since the bound $\alpha(q,N)\ge N\alpha^+_{free}(q)$ would already be sufficient. Yet we decided to include the proof below which works uniformly in $1\leq q \leq 4$ because it highlights well the different boundary conditions at work near the joint line $\Z$ and because the exponent $\alpha_\C^+$ will also play a key role  later (in the proof of the anchoring Lemma \ref{l.anchor}).
\end{remark}

\begin{remark}
In fact, we expect that as soon as percolation occurs in $\mathbb B_N$, then 
\[
\alpha(q,N)=N\alpha^+_{wired}(q)\,.
\]
This comes from the intuition that the infinite cluster at $p_c$ in $\mathbb B_N$ is staying close to the axis, and that this cluster acts as a wiring of vertices. We will turn this intuition into a proof thanks to the {\em random currents} representation in the special case of $q=2$ in Section \ref{s.Ising}. As $\alpha^+_{wired}(q)$ should be equal to $\tfrac2\pi\arccos(\sqrt q/2)$ (see \cite{Smi10,DMT20}), this is consistent with our results (and predictions) on $N^*(q)$ in Theorem \ref{thm:FK*}. 
\end{remark}

\medskip

\begin{proof}[Proof of Theorem \ref{thm:FK} given Proposition \ref{pr.unif}]
To prove Theorem \ref{thm:FK}, it remains to treat the general case where the edge-density on $\mathbb Z$ is an arbitrary number $\lambda\ge0$. (The same argument also applies to the case where edges along $\Z$ have the same weight $p$ as the other edges). Consider $N$ such that $\alpha(q,N)>1$ and $N'$ such that the process given by the pairs of neighboring edges $x$ and $x'$ in $\mathbb Z$ that are connected to each other in $\Book_{N'}$ is dominating a FK percolation of parameter $\lambda^*$ on $\mathbb Z$ (the existence of this integer $N'$ is easy using finite energy). Then, one can easily check that the restriction to $\mathbb B_N$ of FK percolation with parameter $\lambda$ on $\mathbb B_{N+N'}$ is dominating FK percolation on $\mathbb B_N$ with parameters $p_c$ and ${\lambda^*}$. This concludes the proof. 
\end{proof}

\section{Proof of Proposition~\ref{thm:FK2}}\label{sec:3}

\subsection{Preliminaries.}

Let $S\subset \mathbb B_N$. We call a {\em cluster in $S$} a connected component $C\subset S$ of the graph with vertex-set $S$ and open edges with {\em both} endpoints in $S$. We will use the notion of $K$-{\em block} $B^i_K$  to be the translate by the vector $(iK,0)$ of the union, in each page, of the squares $[-K,K)\times[0,K]$. For simplicity we write $B_K$ instead of $B^0_K$. Given a block $B_K^i$, we write $\mathbf C(B^i_K)$ for the cluster {\em in} $B^i_K$ which has the largest intersection with $\Z$ (when there is more than one, pick one according to a deterministic rule). 

We will need the following two definitions.

\begin{definition}[$\theta$-bad block]
A $K$-block $B_K^i$  is $\theta$-{\em good} if $|\mathbf C(B_K^i)\cap \mathbb Z|\ge 2\theta K$.   When a block is not $\theta$-good, we call it $\theta$-{\em bad}. 
Introduce
\[
p_\lambda(K,\theta):=\mathbb P_{B_K,\lambda}^0[B_K\text{ $\theta$-bad}\,].
\] 
\end{definition}

\begin{definition}[bridged block]
A $K$-block $B_K^i$ is {\em bridged in $B_{CK}$} if there exist $-C\le i_-\le i-2$ and $i+2\le i_+\le C$ such that 
\begin{itemize}[noitemsep,nolistsep]
\item $B_K^{i_-}$ and $B_K^{i_+}$ are $\tfrac34$-good.
\item $\mathbf C(B_K^{i_-})$ and $\mathbf C(B_K^{i_+})$ are connected together in $B_{CK}\setminus B_K^i$. \end{itemize}
Introduce
\[
q_\lambda(K,C,i):=\mathbb P_{B_{CK},\lambda}^0[B_K^i\text{ not bridged in $B_{CK}$}].
\]
\end{definition}

\subsection{Bound on $q_\lambda(K,C,i)$.}

The core of the proof of our theorem will be the following proposition.

\begin{proposition}\label{prop:not bridged}
For every $1\le q\le 4$ and $\alpha<\alpha(q,N)$, there exists
 $D_0(\alpha)=D_0(\alpha,q,N)>0$  such that \begin{align}\label{eq:upper q}
q_\lambda(K,C,i)\le \frac{D_0(\alpha)}{(C-|i|)^{\alpha}}~+~2Cp_\lambda(K,\theta)
\end{align}
for every $\lambda\ge p_c$, $N\ge1$, $\theta>\tfrac34$, and $K,C\ge 2$.
\end{proposition}

The proof of Proposition~\ref{prop:not bridged} is divided into two independent lemmata, referred to as the {\em anchoring lemma} and the {\em bridging lemma}.

For $M,K\ge2$, introduce the set $A(M,K)$ to be the union of the half-annulus $\mathbb H\cap\Lambda_{2MK}\setminus\Lambda_{MK}$ and the blocks $B_K^j$ with $j\in(M,2M)$. For a set $\gamma$, introduce the boundary condition $\gamma$ to be the wired boundary condition on $\gamma$, and free elsewhere (see Fig.~\ref{fig:anchoring}).

\begin{figure}[!htp]
\begin{center}
\includegraphics[width=1.00\textwidth]{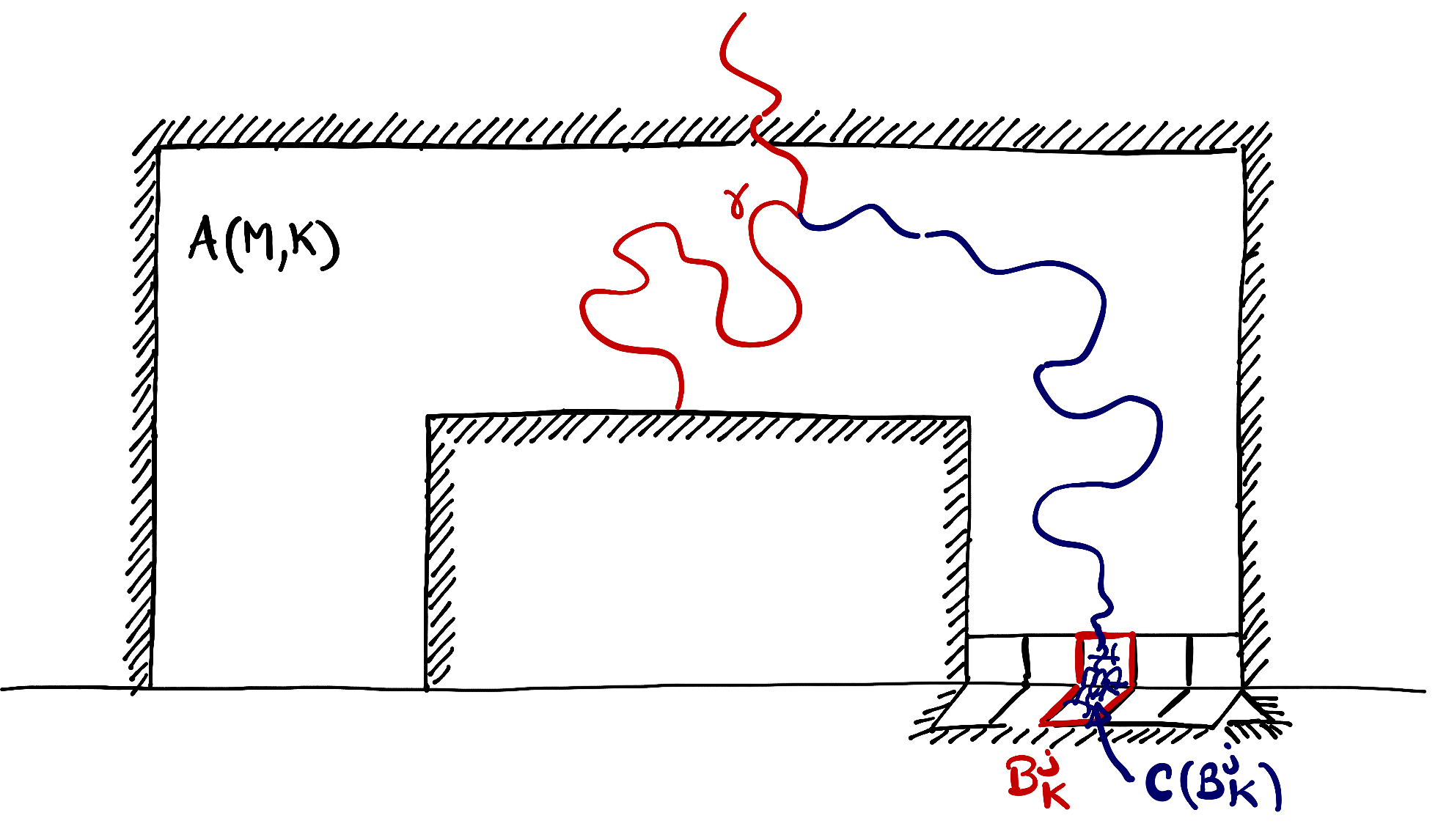}
\end{center}
\caption{A picture of $A(M,K)$ and the path $\gamma$, as well as the event under consideration in the next lemma. The boundary condition $\gamma$ corresponds to wired on the path $\gamma$ and free on the dashed area.}\label{fig:anchoring}
\end{figure}

\begin{lemma}[Anchoring Lemma]\label{l.anchor}
There exists $c_\mathrm{anchor}>0$ such that for every $\lambda\ge p_c$, every integers $K,M$, every $\theta>\tfrac34$, and every path $\gamma$ from $\partial \Lambda_{MK}$ to $\partial\Lambda_{2MK}$ staying above $(0,K)+\mathbb Z$,
\[
\mathbb P_{A(M,K),\lambda}^\gamma[\exists j\in(M,2M): B_K^j\text{ $\theta$-good }\&\,\mathbf C(B_K^j)\leftrightarrow \gamma\text{ in }\mathbb H] \geq c_\mathrm{anchor} (1-p_\lambda(K,\theta))^2.
\]

\end{lemma}

\begin{proof}
Consider the increasing coupling between random-cluster models (see e.g.~\cite{Gri06} for details) $\mathbf P$ between two configurations $\omega'\le \omega$ with 
\[
\omega\sim \mathbb P_{A(M,K),\lambda}^\gamma\quad\text{and}\quad\omega'\sim \mathbb P_{A'(M,K),p_c}^\gamma\quad,\]
 where $A'(M,K)$ is the restriction of $A(M,K)$ to the first two pages (it is a subset of the plane), defined as follows (see for example \cite{DumRaoTas17}). The coupled configuration $(\omega',\omega)$ is written as an increasing function $F$ of i.i.d.~uniform variables in $U_e \in [0,1]$ which are indexed by the edges of $A(M,K)$. To define 
 \[
 F : [0,1]^{E(A(M,K))} \to \{0,1\}^{E(A'(M,K))} \times \{0,1\}^{E(A(M,K))},\] we proceed inductively: the variables $(U_e:e\in A'(M,K))$ are used one at a time to sample $\omega'_e \leq \omega_e$ given the values of the former edges that have been fixed . Once all edges $e\in A'(M,K)$ have been fixed, the remaining variables $(U_e: e\notin A'(M,K))$ are used to sample the remaining edges for $\omega$. 

Define now $\mathbf N$ to be the number of pairs $(j,x)$ with $j\in [\tfrac{5M}4,\tfrac{7M}4]$ and $x\in\mathbb Z$ such that  
\begin{itemize}
\item $B_K^j$ is $\theta$-good in $\omega$;
\item $x\in \mathbf C(B_K^j)(\omega)$;
\item $x$ is connected to $ \gamma$ in $\omega'\cap\mathbb H$.
\end{itemize} 
The fact that $F$ is increasing implies FKG property for $(\omega',\omega)$, which itself gives
  \begin{align*}\label{}
\mathbf E[\mathbf N] & = \sum_{j=5M/4}^{7M/4}\sum_{x\in \mathbb Z} \mathbf P[B_K^j\text{ $\theta$-good in }\omega,x \in \mathbf{C}(B_K^j)(\omega),x\longleftrightarrow \gamma\text{ in }\omega'\cap \mathbb H] \\
& \ge  \sum_{j=5M/4}^{7M/4}\sum_{x\in \mathbb Z} \mathbf P[B_K^j\text{ $\theta$-good in }\omega,x \in \mathbf{C}(B_K^j)(\omega)]\mathbf P[x\longleftrightarrow \gamma\text{ in }\omega'\cap \mathbb H].
\end{align*}
On the one hand, standard crossing estimates and mixing properties of the critical FK percolation with $1\le q\le 4$ give that there exists $c_0>0$ such that 
\[
\mathbf P[x\longleftrightarrow \gamma\text{ in }\omega'\cap \mathbb H]=\mathbb P_{D'}^\gamma[x\longleftrightarrow \gamma\text{ in }\mathbb H]\ge c_0\;  a^+_{\C}(MK)\,.
\]
On the other hand, the definition of $\theta$-good $K$-blocks immediately gives that
\begin{align*}
\sum_{x\in \mathbb Z}\mathbf P[B_K^j\text{ $\theta$-good in }\omega,x \in \mathbf{C}(B_K^j)(\omega)]&=\mathbb E_{B_K^j}[|\mathbf{C}(B_K^j)(\omega)|\mathbbm 1_{B_K^j\text{ $\theta$-good}}]\\
&\ge 2\theta K(1-p_\lambda(K,\theta)).
\end{align*}
Altogether, we deduce the following lower bound on the first moment of $\mathbf N$:
\[
\mathbf E[\mathbf N] \ge c_0\theta MK a^+_{\C}(MK)(1-p_\lambda(K,\theta)).
\]
We now turn to a bound on the second moment. By dropping the first condition, replacing the second by $x\in B_K^j$, and observing that each $x$ belongs to at most 2 blocks, we obtain that 
\begin{align*}\label{}
\mathbf E[\mathbf N^2] & \le 4\sum_{x,y} \mathbb P_{D'}^{\mathrm{mix}}[x,y \longleftrightarrow \gamma\text{ in }\mathbb H].
\end{align*}
A standard application of crossing probabilities and quasi-multiplicativity, see e.g.~\cite{??}, shows that 
\begin{align*}\label{}
\mathbf E[\mathbf N^2] & \le C_0MK \, \sum_{k=1}^{MK} \frac{a^+_{\C}(MK)^2}{a^+_{\C}(k,MK)}\le C_1(MK)^2 \,a^+_{\C}(MK)^2.
\end{align*}
Cauchy-Schwarz inequality implies that the probability that $\mathbf N>0$ is bounded from below by $c_1\theta^2 (1-p_\lambda(K,\theta))^2$. Since $\mathbf N>0$ implies the event under consideration, the claim is proved.
\end{proof}

\begin{remark}
At first sight, a natural way to try proving the Anchoring Lemma would be to run a direct second moment argument on the number, say $\mathbf M$, of points on the middle line $\Z$ which are connected to $\gamma$ in the first page $\H(=\H^1)$ instead of considering the more complicated $\mathbf N$. This works well in the $q=1$ case, but as soon as $q>1$ this strategy seems difficult to implement. Indeed, the first moment $\Eb{\mathbf M}$ would involve in this case the one-arm event in a page $\H$ but for the FK measure in the full book graph $\Book_N$. So far so good, but difficulties arise when controlling $\Eb{\mathbf M^2}$ as a quasi-multiplicativity statement for this arm event would be needed. One way to achieve this would be to prove a version of the \textbf{mixing lemma} (as in \cite{DHN11} in the plane) for the FK measure on the book $\Book_N$. This does not seem straightforward as different pages  may interact via the joint line $\Z$.
This is the reason why we introduce in the proof above a suitable coupling argument in order to transfer the problem to a setting where one can apply a more standard second moment method.
\end{remark}

We now turn to the Bridging lemma.
For integers $K,D,\rho>0$ and a small real number $\eta>0$, set $R_k:=K(2\rho)^k$ and let $F(K,DK,\rho,\eta)$ be the event that there are at least $\eta\log D$ integers $k\ge0$ such that 
 $R_{k+1}\le DK$
and $F(R_k,\tfrac12R_{k+1})$ occurs. 

 \begin{lemma}[Bridging Lemma]\label{pr.OnePagePrecise}
For every $\alpha<\alpha(q,N)$, there exist $\eta=\eta(\alpha)>0$ and an integer $\rho=\rho(\alpha)>0$ such that for every $\lambda\ge p_c$ and $K,D\ge2$ large enough, 
\[
\mathbb P_{B_{DK},\lambda}^0[F(K,DK,\rho,\eta)]\ge 1-\frac1{D^{\alpha}}.
\]
\end{lemma}

\begin{proof}
By monotonicity, it suffices to show the result for $\lambda=p_c$. Fix $\alpha(q,N)>\beta>\alpha$. By definition of $\alpha(q,N)$, there exists $\rho=\rho(\beta)$ such that for every $K$ large enough and $k\ge0$,
\begin{equation}\label{eq:quantitative}
\mathbb P_{B_{R_{k+1}},p_c}^0[F(R_k,\tfrac12R_{k+1})^c]\le \rho^{-\beta}.
\end{equation}

By conditioning on the configuration outside $B_{R_{k+1}}$, the spatial Markov property and the comparison between boundary conditions combined with the previous displayed equation implies that the probability that $F(R_k,\tfrac12R_{k+1})$ occurs is larger than $1-\rho^{-\beta}$.
 In particular, the number of integers $k$ with $R_{k+1}\le DK$ such that $F(R_k,\tfrac12R_{k+1})$ occurs is dominating a binomial random variable $\mathrm{Binom}(n,p)$ with parameters  $n=\lfloor \log_{2\rho}(D)\rfloor -1$ and $p=1-\rho^{-\beta}$. We deduce that for $\eta=\eta(\beta,\rho)>0$ small enough, the probability that there are fewer than $\eta\log D$ such $k$ is smaller than $1/D^{\alpha}$.\end{proof}

We are now ready to dive into the proof of Proposition~\ref{prop:not bridged}.

\begin{proof}[Proof of Proposition~\ref{prop:not bridged}]
Fix $\theta>\tfrac34$ and observe that if $p_\lambda(K,\theta)\ge\tfrac12$ there is nothing to do\footnote{At this stage one may wonder why we put $2Cp_\lambda(K,\theta)$ in the right-hand side of \eqref{eq:upper q} instead of simply $2p_\lambda(K,\theta)$. The reason comes from the conjecture that \eqref{eq:quantitative} can be obtained essentially in terms of the probability of a dual connection with wired boundary conditions on $\mathbb Z$, and that in order to do that, one may want to assume that $p_\lambda(K,\theta)<1/C$. We refer to Sections~\ref{sec:4} and \ref{sec:5} for details of such an application.}. We therefore now assume the opposite. Since the box of size $DK$ around $(Ki,0)$ is included in $B_{CK}$ and being bridged is an increasing event,
the comparison between boundary conditions implies that it suffices to treat the case $i=0$ in the block $B_{DK}$ with $D:=C-|i|$.
 
Fix $\alpha<\alpha(q)$  and consider $\eta=\eta(\alpha)$ and $\rho=
\rho(\alpha)$ given by the Bridging Lemma. 
Also, write $F:=F(K,D,\rho,\eta)$. Thanks to the Bridging lemma and the comparison between boundary conditions,
\[
\mathbb P_{B_{DK},\lambda}^0[F]\ge 1-\frac1{D^{\alpha}}
\]
and it suffices to show that there exists a universal constant $c>0$ such that 
\[
\mathbb P_{B_{DK},\lambda}^0[B_K^i\text{ bridged}|F]\ge 1-\exp[-c\log(D)^2].
\]

\begin{figure}[!htp]
\begin{center}
\includegraphics[width=1.00\textwidth]{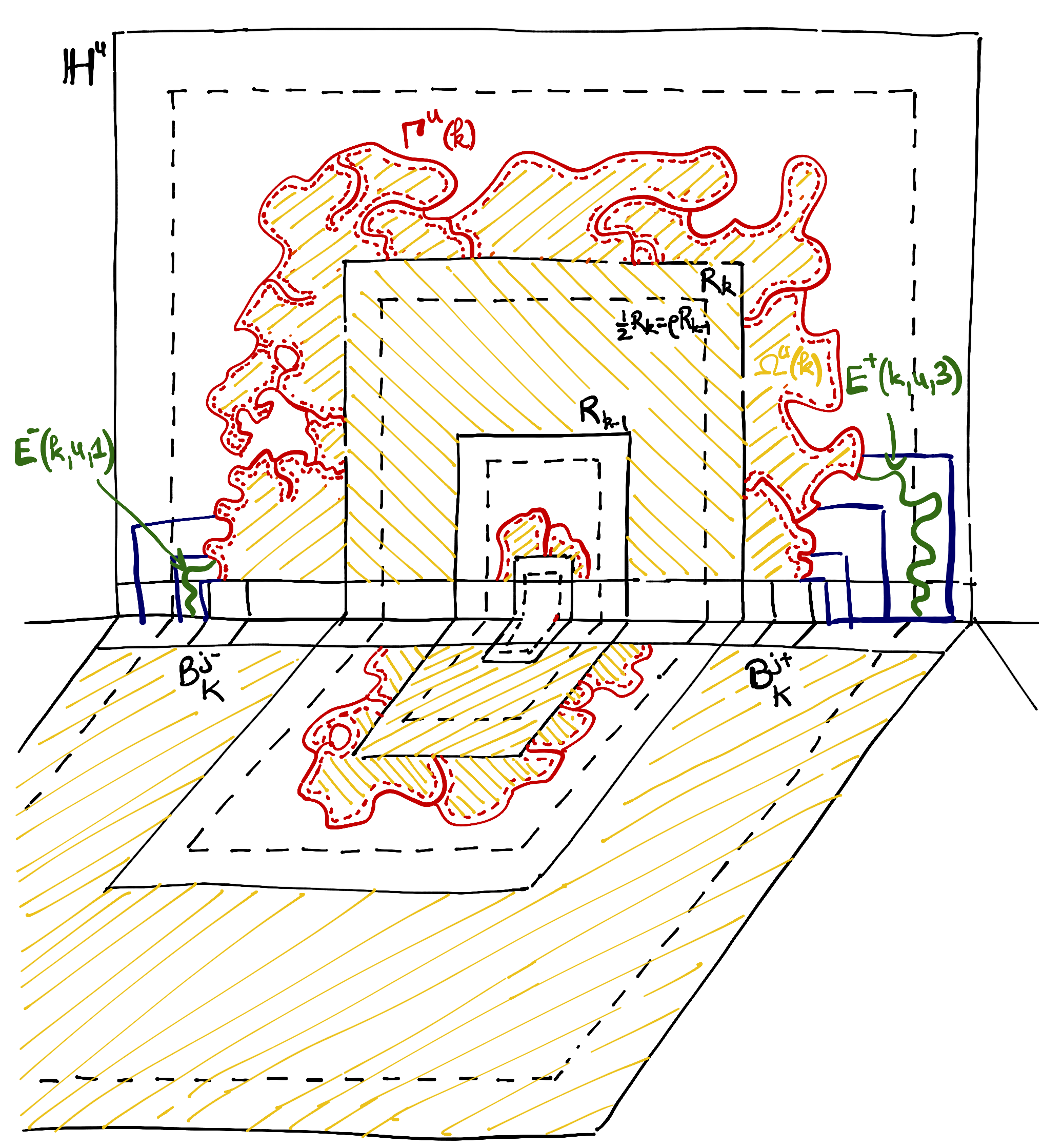}
\end{center}
\caption{A picture of the path $\Gamma^u(k)$ as well as $B_K^{j^\pm}$ and some event $E^+(k,u,3)$ and $E^-(k,u,1)$. The set $\Omega$ is depicted in yellow. Note that these sets do not intersect any of the $B_K^j$ (in other words, they remain at a distance $K$ of $\mathbb Z$).}\label{fig:bridging}
\end{figure}

We now introduce a few quantities (see Fig.~\ref{fig:bridging}). For $k<\lfloor \log_{2\rho}(D)\rfloor$, let $\Gamma(k)$ be the inner-most path in $\omega$ disconnecting $\partial\Lambda_{R_k}$ and $\partial\Lambda_{2^{-1} R_{k+1}}$ in $(0,K)+\mathbb H$ (note that it is a subset of $(0,K)+\mathbb H$). Define $\Omega(k)$ to be the set of $x$ in $((0,K)+\mathbb H)\cap(\Lambda_{2^{-1} R_{k+1}}\setminus\Lambda_{R_k})$ that are surrounding by $\Gamma(k)$, with the convention that the set is $((0,K)+\mathbb H)\cap(\Lambda_{ 2^{-1} R_{k+1}}\setminus\Lambda_{R_k})$ when $\Gamma(k)$ does not exist.  Similarly, define $\Gamma^u(k)$ and $\Omega^u(k)$ as the corresponding quantities in $\mathbb H^u$. 
Finally,  consider the set
\[
\Omega:=\bigcup_{(k,u)}\Omega^u(k).\]
as well as  the set $\mathbf I=\mathbf I(\omega)$ of pairs $(k,u)$ for which $\Gamma^u(k)$ exists, and  the set $\mathbf J=\mathbf J(\omega)$ of triplets $(k,u,i)$ with $(k,u)\in \mathbf I$ and $1\le i<\lfloor \log_2 R_k\rfloor$.

Now, condition on the states of the edges in $\Omega$ and let $\xi$ be the boundary conditions that they induce on $\mathbb B_N\setminus\Omega$. Note that it can be done without revealing any edge outside of $\Omega$ and that $\mathbf I(\omega)$ is measurable in terms of the states of these edges. For each $(k,u)\in \mathbf I(\omega)$,
say that $\Gamma^u(k)$ ends in $B_K^{j_-}$ and $B_K^{j_+}$ on the left and right respectively. For $(k,u,i)\in\mathbf J$, let $E^+(k,u,i)$ be the event that there exists $j$ with $2^{i-1}< j-j_+< 2^i$ such that $B_K^j$ is $\theta$-good and $\mathbf C(B_K^j)$ is connected to $\Gamma^u(k)$ in $\mathbb H^u$. Similarly, define $E^-(k,u,i)$ on the left.
 The comparison between boundary conditions and the anchoring lemma imply that
\begin{align*}
\mathbb P_{\mathbb B_N\setminus \Omega,\lambda}^\xi[B_K\text{ bridged}|F] & \\
 & \hskip - 1cm \ge\mathbb P_{\mathbb B_N\setminus \Omega,\lambda}^\xi[\exists (k,u,i)\in\mathbf J\text{ such that both }E^\pm(k,u,i)\text{ occur}|F]\\
 &\hskip - 1cm  \ge 1-\mathbb E_{\mathbb B_N\setminus \Omega,\lambda}^\xi\Big[\prod_{(k,u,i)\in \mathbf J}\Big(1-\big[c_\mathrm{anchor} (1-p_\lambda(K,\theta))^2\big]^2\Big)\Big|F\Big]\\
 &\hskip - 1cm  \ge 1-(1-c)^{\eta'(\log D)^2},
 \end{align*}
 where in the last line we used the fact that on $F$, $|\mathbf J(\omega)|\ge\eta'(\log D)^2$, and that the assumptions that $\theta>\tfrac34$ and $p_\lambda(K,\theta)\le \tfrac12$ guarantee the existence of $c>0$.
 \end{proof}

\subsection{Proof of Proposition~\ref{thm:FK2}.}

The proof of Proposition~\ref{thm:FK2} relies on the idea that clusters at scale $K$ and local density $\theta$ will merge and with high probability create  new clusters at scale $CK$ of local density $\theta'=\theta-O(1/C)$ slightly smaller than $\theta$ (this slight loss of density allows us to lose a few clusters at scale $K$ in the process). More precisely, we prove the following renormalization inequality.

\begin{lemma}\label{lem:1} Let $N>1$ such that $\alpha(q,N)>1$ and $\theta>\tfrac34$. There exist $C_0\ge1$  large enough (depending on $\theta$ and $N$) such that the following holds. 
For every $\lambda>0$ and for every integers $C\ge C_0$ and $K\ge 2$,
\begin{equation}\label{eq:1}
p_\lambda(CK,\theta - C_0/C)\le \tfrac1{100}\,p_\lambda (K,\theta)+ 6C^2\,p_\lambda(K,\theta)^2.
\end{equation}
\end{lemma}

\begin{proof}
Fix $1<\alpha<\alpha(q,N)$. Let $C_0>0$ be a large constant to be chosen later and set $\theta':=\theta-C_0/C$.  For $|i|\le C$, let $E_i$ be the event that $B_K^i$ is $\theta$-bad and all the blocks $B_K^j$ are $\theta$-good for $j\in [-C,C]\setminus\{i-1,i,i+1\}$, and set
  \begin{equation}
    \label{eq:9}
    F_i=E_i\cap \{B_{CK}\text{ is $\theta'$-bad}\}.
  \end{equation}
   Observe that if all $K$-blocks $B^j_K$,  $-C\le j \le C$,  are $\theta$-good, then the assumption that $\theta>3/4$ imposes that all the clusters $\mathbf C(B^j_K)$ are connected together in $B_{CK}$, which implies the existence of a cluster in $B_{CK}$ with cardinality larger than $2 \theta CK$. In particular, if $B_{CK}$ is $\theta'$-bad, then either there exist two disjoint $\theta$-bad $K$-blocks, or there exists $i$ such that $E_i$ occurs.  The union bound implies
  \begin{equation}
    \label{eq:6}
    p_\lambda(CK,\theta')\le \sum_{i=-C}^C\mathbb P_{B_{CK},\lambda}^0[F_i]+ \mathbb P[\text{there are at least two {\em disjoint} $\theta$-bad $K$-blocks}].
  \end{equation}
By the spatial Markov property and the comparison between boundary conditions, we have
\begin{equation}\label{eq:1a}
\mathbb P_{B_{CK},\lambda}^0[\text{there are at least two {\em disjoint} $\theta$-bad $K$-blocks}]\le  \binom{2C-1}{2}p_\lambda(K,\theta)^2.
\end{equation}

  It remains to bound the first term in \eqref{eq:6}, which is the object of the end of the proof. If all $K$-blocks $B^j_K$ with $|j|\le C-C_0$  are $\theta$-good, the same argument as above implies that $B_{CK}$ is $\theta'$-good, therefore $F_i=\emptyset$ whenever $|i| \ge C-C_0$. 
   Now, let $|i|\le C-C_0$. 
   Note that if $B_K^i$ is bridged in $B_{CK}$, then $B_{CK}$ is also $\theta'$-good. Furthermore, when $B_K^i$ is not bridged (this event does not depend on edges in $B_K^i$), for $E_i$ to occur then $B_K^i$ must be $\theta$-bad. As a consequence, the spatial Markov property and the comparison between boundary conditions implies that
    \begin{align}
   \mathbb P[F_i]&\le  \mathbb P[E_i|B_K^i\text{ not bridged}]\mathbb P[B_K^i\text{ not bridged}]\nonumber\\
   &\le p_{\lambda}(K,\theta) q_{\lambda}(K,C,i)\nonumber\\
   &\le \frac{D_0(\alpha) \, p_{\lambda}(K,\theta)}{(C-|i|)^{\alpha}} +2Cp_{\lambda}(K,\theta) ^2, \label{eq:10}
 \end{align}
 where in the last line we invoked Proposition~\ref{prop:not bridged} for $\alpha$. 
Select $C_0$ so large that
\[
\sum_{|i|\le C-C_0}\frac{2 D_0(\alpha)}{(C-|i|)^{\alpha}}\le \tfrac1{100}.
\]
Plugging \eqref{eq:10} and \eqref{eq:1a} in \eqref{eq:6} concludes the proof.\end{proof}

\begin{proof}[Proof of Proposition~\ref{thm:FK2}]
Let $N$ satisfying $\alpha(q,N)>1$. Choose $\theta_1<1$ and $C_1\ge C_0$ (where $C_0$ is provided by Lemma~\ref{lem:1}) such that the sequences
\[
\begin{cases}C_{n+1}=(n+1)^{3} C_1,\\
\theta_{n+1}:=\theta_n- \tfrac{C_0}{C_{n+1}},\end{cases}\quad\text{ for $n\ge1$}
\]
satisfy $\theta_n>\tfrac34$ for every $n\ge1$.
Now, set $\lambda^*\ge p_c$ so large that 
\[
  p_{\lambda^*}(C_1,\theta_1)\le \mathbb P_{\lambda^*}[\exists \{x,x+1\}\subset B_{C_1}\cap \mathbb Z\text{ closed}]\le C_1\frac{1-\lambda^*}{q-(q-1)\lambda^*}\le \frac1{1200 C_1^2}. 
\]
and consider the sequence of scales defined\footnote{Note that it gives  $K_n=(n!)^3 C_1^n$ for all $n\geq 1$.} by
\begin{equation}
  \label{eq:13}
  \begin{cases}
    K_1=C_1,& \\
    K_{n+1}=C_{n+1} K_n&n\ge1.
  \end{cases}
\end{equation}
Applying Lemma~\ref{lem:1} to $(N,\theta_{n},C_n,K_n)$, we see that the sequence $u_{n}:=p_{\lambda^*}(K_n,\theta_n)$ satisfies 
\[
  \forall n\ge 1,\qquad u_{n+1}\le \tfrac1{100} u_n+6 C_n^2 u_n^2.\]
By induction, we obtain that $u_n\le  \tfrac1{1200} C_n^{-2}$ for every $n\ge 1$, and  therefore,
\[
\mathbb P_{\lambda^*}[B_{K_n}\text{ $3/4$-good}]\ge 1 - u_n\ge 1-\tfrac1{1200}C_n^{-2}\ge \tfrac12.\] 
First using the estimate above and then  translation invariance, we get that  for every~$n\ge1$, 
\begin{align}
  \tfrac34 K_n&\le\mathbb E_{\lambda^*}[|\mathbf C(B_{K_n}) \cap \Z|1_{B_{K_n} \text{ $3/4$-good}} ]\nonumber\\
  &\le 2K_n \mathbb P_{\lambda^*}[0\text{ is in a cluster of size at least $\tfrac32 K_n$}].  \label{eq:14}
\end{align}
Dividing both sides by $2K_n$, we obtain
\[
\mathbb P_{\lambda^*}[0\text{ is in a cluster of size at least $\tfrac32 K_n$}]\ge \tfrac 38,
\]
which by measurability implies that the probability that 0 is connected to infinity is larger than or equal to $\tfrac38$.
\end{proof}

\section{Proof of Theorem \ref{thm:FK*} for $q\ne 2$}\label{s.2dM}\label{sec:4}

In this section, we prove the following two claims of Theorem \ref{thm:FK*}: first we show that $N^*(q)\geq 3$, when $q\in[1,2)$ and second, we prove that $N^*(q)=2$ for all $2<q<4$ and that $N^*(4)\leq 2$. The more subtle case of  $q=2$ will be analyzed in the next section with the help of random currents. 

We start with the following proposition corresponding to the first claim.

\begin{proposition}
For every $1\le q<2$, there exists $c=c(q)>0$ such that for every $n\ge1$ and $\lambda\in(0,1)$,
\[
\P^1_{\mathbb B_2,\lambda}[\Lambda_n\text{ horizontally crossed}]\le 1-(1-\lambda)c.
\]
In particular, 
$\P^1_{\mathbb B_2,\lambda}[0\longleftrightarrow\infty]=0$. 
\end{proposition}

The proof is based on a second-moment argument.

\begin{proof}
Define the number $\mathbf N$ of edges $e\subset [-n/2,n/2]$ such that the endpoints of $e^*$ are respectively dual connected to the top of $\Lambda_n$ in the upper half-plane, and to the bottom of $\Lambda_n$ in the lower half-plane. Under $\P^1_{\mathbb B_2,1}$ (for which $\Z$ is completely wired), both pages behave independently and we immediately get that
\[
\E^1_{\mathbb B_2,1}[\mathbf N]\ge c_0na_{wired}^+(n)^2.
\]
In the other direction, the second moment gives, using classical quasi-multiplicativity estimates
\[
\E^1_{\mathbb B_2,1}[\mathbf N^2]\le C_0n\sum_{k=1}^n\frac{a_{wired}^+(n)^4}{a_{wired}^+(k,n)^2}\le C_1n^2a_{wired}^+(n)^4,
\]
where in the last inequality we used a result from \cite{DMT20} stating the existence of  $c_1=c_1(q)>0$ such that for every $k\le n$,
\begin{equation}\label{eq:estimate two arm q<2}
a_{wired}^+(k,n)\ge c_1(\tfrac kn)^{1/2-c_1}.
\end{equation}
Overall, we get by comparison between boundary conditions and Cauchy-Schwarz that
\[
\P^1_{\mathbb B_2,\lambda}[\mathbf N>0]\ge \P^1_{\mathbb B_2,1}[\mathbf N>0]\ge c_2.
\]
Now, on $\{\mathbf N>0\}$ (which does not prescribe anything on edges in $\Z$), $\Lambda_n$ is not crossed horizontally if any of the edges of $\Z$ such that the endpoints of $e^*$ are dual-connected to top and bottom is in fact closed. Since there is at least one such edge, we get that
\[
\P^1_{\mathbb B_2,\lambda}[\Lambda_n\text{ horizontally crossed}]\le 1-c_2(1-\lambda).
\]
This concludes the proof of the first part of the proposition. For the second part, ergodicity gives that $\P_{p_c,\lambda,q}[0\longleftrightarrow\infty]=0$ since otherwise the crossing probability would tend to 1.
\end{proof}

We now turn to the other claim, which we split in two.

\begin{proposition}\label{}
For any $2<q\leq 4$, we have $N^*(q)\leq 2$. 
\end{proposition}

\begin{proof}
We wish to prove that for $\lambda>0$ large enough, $\P^1_{\mathbb B_2,\lambda}[0\leftrightarrow\infty]>0$. In order to do that, we only need to prove the equivalent of Proposition~\ref{prop:not bridged}, i.e. that for some constant $\alpha>1$, there exists
 $D_0(\alpha)>0$  such that \begin{align}\label{eq:upper q>2}
q_\lambda(K,C,i)\le \frac{D_0(\alpha)}{(C-|i|)^{\alpha}}~+~2Cp_\lambda(K,\theta)
\end{align}
for every $\lambda\ge p_c$, $N\ge1$, $\theta>\tfrac34$, and $K,C\ge 2$. 

To do that, observe that for $B_K^i$ not to be bridged in $B_{CK}$, 
\begin{itemize}
\item either there must be a $\theta$-bad box $B_K^j$, an event which occurs with probability smaller than $2Cp_\lambda(K,\theta)$, 
\item or all the boxes are $\theta$-good, in which case if $\mathbf C$ is the cluster gathering all the $\mathbf C(B_K^j)$, we have that $B_K^i$ is dual connected to $\partial B_{CK}$ above and below $\mathbf C$.\end{itemize}
Yet, when working on $\mathbb B_2=\Z^2$, one notices that $\mathbf C$ contains a crossing from left to right in $[-CK,CK]\times[-K,K]$. In particular, conditioned on the bottom-most such crossing $\Gamma$ and everything below it, the spatial Markov property together with the comparison between boundary conditions of the model imply that the probability that there exists a dual path from $B_K^i$ to $\partial B_{CK}$ above $\Gamma$ is bounded by 
the probability that there exists a dual-connected path in $[-CK,CK]\times[-K,CK]$ from $B_K^i$ to $\partial B_{CK}$, with free boundary conditions on $\partial B_{CK}$ and wired on the bottom. In particular, it is bounded by $C_0a_{wired}(K,(C-|i|)K)$ using classical mixing properties coming from \cite{DumSidTas13}. 
Now, it was proved in \cite{DMT20} that for every $q<2\le 4$ (it was not done for $q=4$ but the same proof extends), there exists $c_1=c_1(q)>0$ such that for every $k\le n$,
\begin{equation}\label{eq:estimate two arm q>2}
a_{wired}^+(k,n)\le \tfrac1{c_1}(\tfrac kn)^{1/2+c_1}.
\end{equation}

Altogether, we deduce that for some constant $D_1>0$, 
\[
\mathbb P_{\mathbb B_2,\lambda,q}^1[B_K^i\stackrel{*}{\leftrightarrow}\partial B_{CK}\text{ above }\mathbf C|B_K^j\text{ all good},B_K^i\stackrel{*}{\leftrightarrow}\partial B_{CK}\text{ below }\mathbf C]\le\frac{D_1}{(C-|i|)^{1/2+c_1}}.
\]
Similarly, one proves that 
\[
\mathbb P_{\mathbb B_2,\lambda,q}^1[B_K^i\stackrel{*}{\leftrightarrow}\partial B_{CK}\text{ below }\mathbf C|B_K^j\text{ all good}]\le\frac{D_1}{(C-|i|)^{1/2+c_1}}.
\]
Combining this two displayed inequalities with the two bullets above gives \eqref{eq:upper q>2} for $\alpha:=1+2c_1$, a fact which concludes the proof.
\end{proof}

\begin{proposition}\label{}
For any $2<q<4$, we have $N^*(q)\ge2$. 
\end{proposition}

\begin{proof}
The lower bound $N^*(q)\ge2$ is a straightforward consequence of the strong RSW Theorem from \cite{DumSidTas13}, that implies that on $\mathbb H$, the probability that there exists a dual path from $[-2n,-n]$ to $[n,2n]$ surrounding $\Lambda_n$ in $\Lambda_{2n}$ is bounded from below by a constant $c_0>0$. This contradicts the fact that this probability should tend to 0 for $0$ to be connected to infinity.
\end{proof}

\section{Book-Ising with three pages and random currents}\label{s.Ising}\label{sec:5}

The purpose of this section is to show that a first-order phase transition already arises with only $3$ pages for Book-Ising ($q=2$). This corresponds to $N^*(q=2) \leq 3$ and our proof is consistent with the prediction from \cite{SMP10,Ste14}.
 The main technique will involve random currents. To highlight the main ideas, before handling the graph $\Book_{3}$, we will start in the subsection below with an interesting question on its own where a positive (i.i.d) density of sites along the middle line $\Z\subset \Book_3$ are oriented in the $+$ direction. We will only give a short sketch of proof for the toy-model and will leave the detailed proof to the  true Book-Ising (as such the former may be viewed as an outline of proof of the second in a simpler setting).

\subsection{A positive density of $+$ is indistinguishable from a $+$ boundary condition.}

In this subsection, we will give a sketch of proof of the following result: the decoupling property from Theorem \ref{thm:Ising2} holds in the simpler setting of the half-plane where a positive density of sites on $\mathbb Z$ are wired together. This will serve as a useful toy model for Theorem \ref{thm:Ising2}. The reader comfortable with the random current terminology may skip this section if needed. The statement above is a 2D version of a result by  Bodineau \cite{Bod06}. 

Let us set some notations: $\rho\in[0,1]$ will denote the bias of our Bernoulli quenched disorder along the line $\Z\times \{0\}$. Let $\eta\sim \mathrm{Bernoulli}(\rho)^{\otimes \Z}$ be i.i.d.~Bernoulli random variables attached to each site $i$ on the middle line $\Z=\Z\times \{0\}$. Given $\eta$, $\mathbb H^\eta$ will denote the (random) graph where all points $i\in \Z$ for which $\eta_i=1$ are wired together. 
Finally, $\mu_{\H^\eta,\beta_c}^+$ denotes the Ising measures on $\H^\eta$ with  $+$ boundary conditions. 
\begin{theorem}\label{th.BookFactor}
For any $0<\rho<1$ and any $m, L \in \N$, there exist $c,C\in(0,\infty)$ so that for 
any set $A$ made of $m$ vertices at a distance at least  $L$ of  $\mathbb Z$, then with probability at least $1-L^{-C}$ in the quenched disorder $\eta$, we have  
\begin{align*}
\mu_{\H^\eta,\beta_c}^+[ \sigma_A]&= \mu_{\H,\beta_c}^{+}[\sigma_A]
(1-O_m(L^{-c}))\,.
\end{align*}
\end{theorem}

\begin{remark}
Note that by Griffiths inequalities, we deduce that if we consider the graph $\mathbf B_3^\eta$ constructed like $\H^\eta$ but from $\mathbb B_3$ instead of $\mathbb H$, we immediately get that for every three sets $A\subset \H^1$, $B\subset\H^2$ and $C\subset H^3$ made of vertices all at a distance at least $L$ from $\Z$, we have that with probability at least $1-L^{-C}$ in the quenched disorder $\eta$, we have that
\begin{align*}
\mu_{\mathbb B_3^\eta,\beta_c}^+[ \sigma_A\sigma_B\sigma_C]&=\mu_{\H,\beta_c}^{+}[\sigma_A]\mu_{\H,\beta_c}^{+}[\sigma_B]\mu_{\H,\beta_c}^{+}[\sigma_C](1-O_m(L^{-c}))\,.
\end{align*}
\end{remark}

The proof of the theorem  requires the introduction of another representation, called the {\em random-current representation}. We refer to \cite{Aiz82,Dum16} for details on this representation and briefly define it here. A {\em current} $\n$ on $G=(V,E)$ is a function from $E$ to $\mathbb N:=\{0,1,2,\dots\}$.    A {\em source} of $\n=(\n_{xy}:xy\in E)$ is a vertex $x$ for which $\sum_{y\sim x}{\n}_{xy}$ is odd. The set of sources of $\n$ is denoted by $\partial\n$. Also set
$$w_\beta(\n):=\prod_{xy\in E}\frac{\displaystyle\beta^{\n_{xy}}}{\n_{xy}!}.$$
Currents are useful as they lead to the following expression for spin-spin correlations:
\begin{equation}\label{eq:erg}
\mu_{G,\beta}^f[\sigma_x\sigma_y]=\frac{\displaystyle\sum_{\partial\n=\{x,y\}}w_\beta(\n)}{\displaystyle\sum_{\partial\n=\emptyset}w_\beta(\n)}.
\end{equation}
For more general spin-observable $\sigma_A$, $A\subset V$, one has the expression
\begin{align*}\label{}
\mu_{G,\beta}^f[\sigma_A]=\frac{\displaystyle\sum_{\partial\n=A}w_\beta(\n)}{\displaystyle\sum_{\partial\n=\emptyset}w_\beta(\n)}\,.
\end{align*}

Also, a classical use of the {\em switching lemma} enables one to compare the spin-spin correlations on two graphs $H\subset G$ as follows:
\begin{equation}\label{eq:crucial RC}
\mu_{H,\beta}^f[\sigma_x\sigma_y]=\mu_{G,\beta}^f[\sigma_x\sigma_y]\mathbf P^{\{x,y\}}_{G,\beta}\otimes\mathbf P^\emptyset_{H,\beta}[\exists\text{ path of }\n_1+\n_2>0\text{ in }H\text{ from $x$ to $y$}],
\end{equation}
where  $\mathbf P^A_{G,\beta}$ attributes a weight proportional to $w_\beta(\n)$ if $\partial\n=A$ and $0$ otherwise, and the sign $\otimes$ means that we take the product measure (see for example  \cite[Lemma 2.2]{ADCS15}).
For a general spin-observable $\sigma_A$, with $A\subset H$, the identity becomes 
\begin{equation*}
\mu_{H,\beta}^f[\sigma_A]=\mu_{G,\beta}^f[\sigma_A]\mathbf P^{A}_{G,\beta}\otimes\mathbf P^\emptyset_{H,\beta}[\widehat{\n_1+\n_2}_{\md H} \in  \calF_A],
\end{equation*} 
where $\widehat\n\in\calF_A$ is the event that any cluster of $\n>0$ intersecting $A$ must intersect $A$ at an even number of points.  

We shall also need these expressions in the case of a graph  $G$ with $+$ boundary conditions. This means that some points (called the {\em boundary} of $G$) are connected to an extra vertex called the {\em ghost} vertex. (See \cite{Dum16} for a detailed exposition).  In such a case, the last expression for example reads as follows:
\begin{equation*}
\mu^+_{H,\beta}[\sigma_A]=\mu^+_{G,\beta}[\sigma_A]\mathbf P^{A}_{G^+,\beta}\otimes\mathbf P^\emptyset_{H^+,\beta}[\widehat{\n_1+\n_2}_{\md H} \in  \calF_A],
\end{equation*} 
where currents $\n_1,\n_2$ under $\P_{G^+,\beta}^A$ and $\P_{H^+,\beta}^\emptyset$ are now allowed to go through the ghost vertex and their boundary $\p\n_1,\p\n_2$ is only considered on all vertices but the ghost (also $\calF_A$ is now the event that all the clusters which intersect $A$ are either connected to the ghost or intersect $A$ at an even number of points). 

\begin{proof}[Sketch of proof of Theorem~\ref{th.BookFactor}]
We will show that there exists $c\in(0,\infty)$ such that for every $A$ made of $m$ vertices at a distance at least $L$ from $\Z$, 
\[
\E\big[\mu_{\H^\eta,\beta_c}^{+}[ \sigma_A ]\big] \geq 
\mu_{\H,\beta_c}^{+}[ \sigma_A ](1-O_{m}(L^{-c})).
\]
For any large $M$, let $H^{\eta,+}$ be the finite random graph obtained from $\H^\eta$ by connecting all the vertices $i\in \p \H\cap \Lambda_M$ which are such that $\eta_i=1$ to the ghost. We have 
\[
H^{\eta,+}\subset G^+:=\H \cap \Lambda_M \text{ with all points in $\p \H$ connected to the ghost}\,.
\]
Applying the above formula, it remains  to bound from below (for any large $M$) the following average with respect to $\eta$:
\begin{align*}\label{}
\E\big[ \mathbf P^{A}_{G^+,\beta}\otimes\mathbf P^\emptyset_{H^{\eta,+},\beta}[\widehat{\n_1+\n_2}_{\md H^{\eta,+}} \in  \calF_A] \big]\,.
\end{align*}
\begin{figure}[!htp]
\begin{center}
\includegraphics[width=\textwidth]{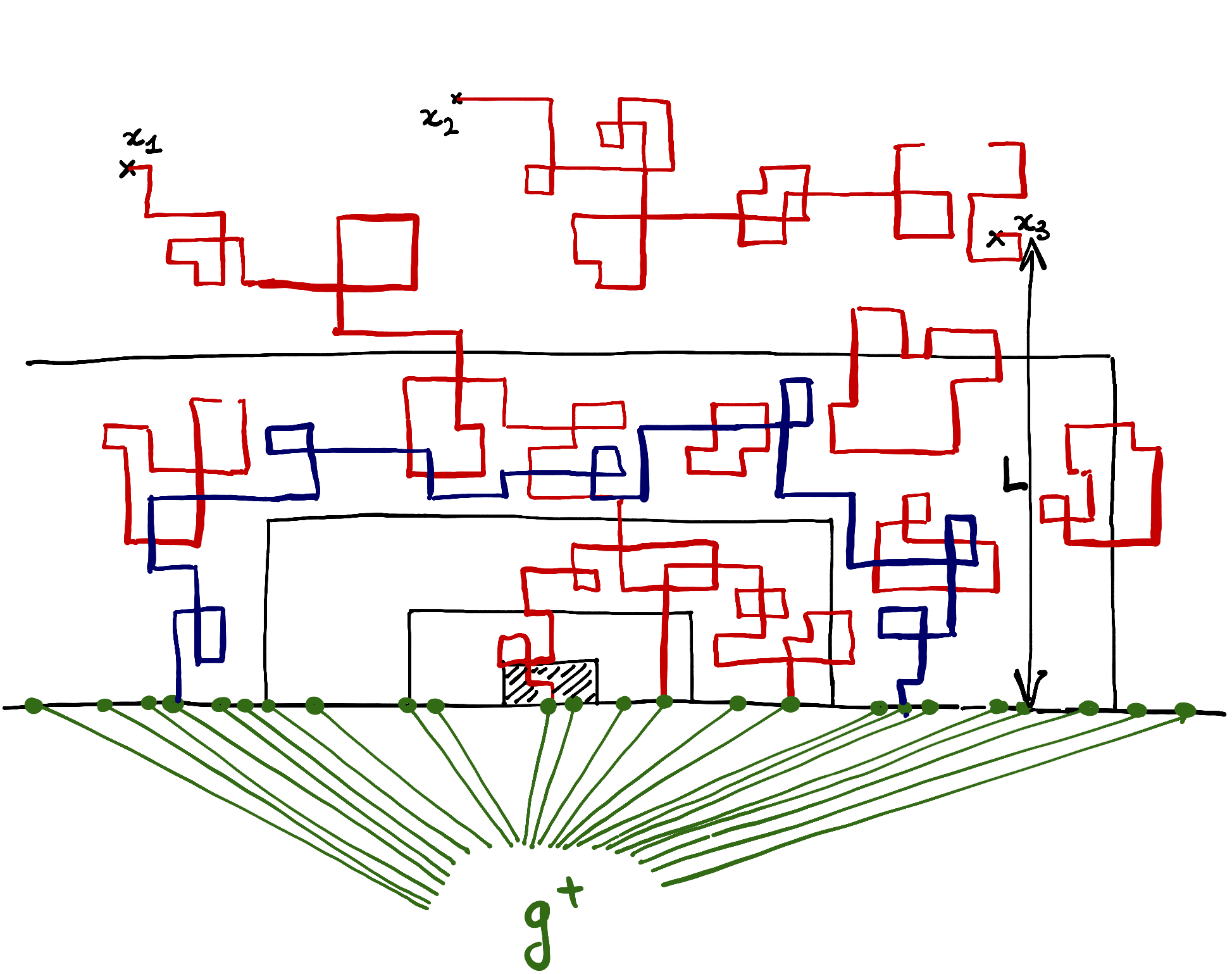}
\end{center}
\caption{The example of the graph $\H^{\eta,+}$ with the additional connections to the ghost (in green). We also depicted the current $\n_1$ in red, as well as a bridge of $\n_2$ (in blue) guaranteeing that a connection in $\H^\eta$ to the ghost. }\label{f.sketch1}
\end{figure}
We proceed as follows:
\bnum
\item Sample $\n_1$ according to $\mathbf P^{A}_{G^+,\beta}$. If $\n_1$ restricted to $G\setminus\H$ is already in $\calF_A$, then whatever $\eta$ and  $\n_2$ are, we must have  $\widehat{\n_1+\n_2}_{\md H^{\eta,+}} \in  \calF_A$. 
\item Suppose then that this is not the case. Then, there is at least one cluster of $\n_1$ intersecting $A$ and reaching $\p \H$ (and there are of course at most $m=|A|$ such clusters). By union bound, let us focus on the case of only one point, say $x$, and let us assume that $\n_1$ connects $x$ with $\p \H$. Among all points in $\Z$ which are connected to $x$ via $\n_1$, choose $u=(k,0)$ to be, say, the furthest on the left and to lighten the notation assume it is equal to the origin.  Notice that $u$ is measurable with respect to $\n_1$ and that neither $\eta$ nor $\n_2$ have been sampled yet.  
\item Let then sample $\eta$ so that the subgraph $H^{\eta,+}\subset G^+$ is now well defined. 
\item In order to match with the setup of Lemma \ref{lem:estimate one arm} below, set $r:=L^{1/2}$, $s:=L^{1/4}$ and $R:=L \leq \dist(x,\Z)$. 
As $k$ has been localized before sampling $\eta$, we can claim that with probability at least $1-\exp(-c s)$, $\eta$ will be sufficiently dense in each $s$-interval  included in  $[k-R,k-r]\cup[k+r,k+R]$. (We will be more explicit in the proof of Lemma \ref{lem:estimate one arm}). 
\item Assuming $\eta$ is sufficiently dense around $u$, the rest of the proof consists in showing that with probability at least $1-r^{-c}$, the current $\n_2$ will create a bridge from $\{i\in [k-R,k-r], \eta_i=1\}$ to $\{i\in[k+r,k+R], \eta_i =1\}$ which, by planarity, will necessarily intersect $\n_1$. As such $x$ will be connected to the ghost via $\widehat{\n_1+\n_2}_{\md H^{\eta,+}}$ as desired. 
\item Finally, the proof that the above bridging property holds with probability $1-r^{-c}$ relies on a coupling between sourceless random-current and FK percolation with parameter $q=2$. This coupling will be  described before the proof of Lemma \ref{lem:estimate one arm}; see Fig.~\ref{f.sketch1}.
\enum 
\end{proof}

We also discuss a slightly more difficult theorem, but which is closer to the one in the next section. Let  $\P_{\H^\eta}^1$ be the FK percolation measure where all the sites $\{ i\in \Z, \eta_i=1\}$ are wired together.   

\begin{theorem}\label{th.Hbernoulli}
For any $0<\rho<1$ and any $1\leq r \leq R$, there exist $c,C\in(0,\infty)$ so that with probability at least $1-R\exp(-c\, r^{-1/2})$ in the quenched disorder $\eta$, we have 
\begin{align}
\FK{\H^\eta}{1}{\Lambda_r\stackrel*{\longleftrightarrow}\partial\Lambda_R} \leq C\big[ (\tfrac rR)^{1/2} + r^{-c}\big]\,. 
\end{align}
\end{theorem}

\begin{proof}[Steck of proof of Theorem \ref{th.BookFactor}]
Define the subsets of $\Z$:
\begin{align*}\label{}
&\mathbf I^-:=\{i\in[-R,-r]: \eta_i=1\}\subset [-R, -r]=:\mathbf J^-, \\ 
&\mathbf I^+:=\{i\in[r, R]: \eta_i=1\}\subset[r,R]=:\mathbf J^+.
\end{align*}

Let $G^\eta$ be the finite graph $\Lambda_R \setminus \Lambda_r$ in which all vertices in $\mathbf I^-$ are connected (wired) to a ghost vertex $g^-$ and all vertices in $\mathbf I^+$ are connected (wired) to a different ghost vertex $g^+$ and where the rest of of the boundary is {\em free}. Let also $\bar G$ be the graph where all vertices in $\mathbf J^-$ are connected  to $g^-$ while all vertices in $\mathbf{J^+}$ are connected to $g^+$ and where the rest of the boundary is {\em free}. The monotony properties of FK percolation and the Edwards-Sokal coupling give that 
\begin{align*}\label{}
&1-\FK{\H^\eta}{1}{\Lambda_r\stackrel*{\longleftrightarrow}\partial\Lambda_R}  \geq 
1-\FK{G^\eta}{}{\Lambda_r\stackrel*{\longleftrightarrow}\partial\Lambda_R}= \mu_{G^\eta,\beta_c}[ \sigma_{g^+} \sigma_{g^-}].
\end{align*}
Notice that we have $G^\eta \subset \bar G$. Similarly as in the above proof, we may now use the {\em switching lemma} via the above identity~\eqref{eq:crucial RC} to obtain
\begin{align*}\label{}
&\mu_{G^\eta,\beta_c}[ \sigma_{g^+} \sigma_{g^-} ] \\
&=
\mu_{\bar G,\beta_c}[ \sigma_{g^+} \sigma_{g^-} ]
\mathbf P^{\{g^+,g^-\}}_{\bar G,\beta_c}\otimes\mathbf
P^\emptyset_{G^\eta,\beta_c}
[\exists\text{ path of }\n_1+\n_2>0\text{ in } G^\eta \text{ from $g^+$ to $g^-$}
]\,.
\end{align*}
From now on, the proof can be concluded in two steps.
\bnum
\item It can be extracted for instance from \cite{DHN11} that 
$\mu_{\bar G,\beta_c}[ \sigma_{g^+} \sigma_{g^-} ] \geq 1 - C( \tfrac r R)^{1/2}\,.
$
\item 
The second step is very similar to the argument outlined above: under $\mathbf P^{\{\sigma^+,\sigma^-\}}_{\bar G,\beta_c}$, $\n_1$ will connect at least a point $u^-\in \mathbf J^-$ to a point $u^+ \in \mathbf J^+$. The goal is thus to show that $\n_2$ will create with high probability $1-Cr^{-c}$ a {\em bridge} which will connect $\mathbf I^-$ and $\mathbf I^+$ to $\n_2$. To prove this, we argue as above: we set an intermediate scale $s:=r^{1/2}$ and we claim that with probability at least $1-R\exp(-c\, r^{1/2})$, the sets $\mathbf I^-=\mathbf I^-(\eta)$ (resp $\mathbf I^+$) will be sufficiently dense at scale $s$ to create many {\em bridges} thanks a coupling between sourceless random-current and FK percolation with parameter $q=2$ (see the proof of Lemma~\eqref{lem:estimate one arm} for a detailed proof). 

At this stage, there is a subtle but important difference compared to the argument in the previous sketch of proof. It could be that the points $u^-$ and $u^+$ could be close to $\partial\Lambda_r$ or $\partial\Lambda_R$. In such case, one cannot really use the argument described above. Yet, it can be proved in this case that the probability that $\n_1$ itself is connected to a vertex of $\Z$ close (say at a distance at most $r^{3/4}$) to $\partial\Lambda_r$ or $\partial \Lambda_R$ is bounded by $Cr^{-c}$ (see the next section for details).
\enum

\end{proof}

\subsection{Proof of Theorem \ref{thm:Potts*}(i)}

The core of the proof will be the following result.

\begin{proposition}\label{prop:technical}
There exist $c_0,C_0\in(0,\infty)$ such that for every $r$ and $R$ such that $r$ divides $R$, every $\lambda\ge p_c$, every $K\ge2$, and every $\theta>\tfrac34$,
\[
\mathbb P_{B_{KR}}^0[F(Kr,KR)^c]\le C_0\Big[\big(\frac rR\big)^{N/2}+r^{-c_0}+R\exp(-c_1\sqrt r)\Big]+2R\,p_\lambda(K,\theta).
\]
\end{proposition}

We start by explaining how to adapt the proof of Theorem~\ref{thm:FK} using Proposition~\ref{prop:technical} to obtain   Theorem~\ref{thm:Potts*}(i).

\begin{proof}[Proof of Theorem~\ref{thm:Potts*}(i)]
Fix $K$ and $C$ and assume for a moment that $K$ is chosen so that $p_\lambda(K,\theta)\le 1/(2C)$. When applied to $R_k:=K(2\rho)^k$, we see that the previous proposition implies that 
\begin{align*}
\mathbb P_{B_{R_{k+1}}}^0&[F(R_k,\tfrac12R_{k+1})^c]\\
&\le C_0\Big[\big(\frac {2R_k}{R_{k+1}}\big)^{N/2}+\big(\frac K{R_k}\big)^{c_0}+\frac{R_{k+1}}{K}\exp\big(-c_1\big(\frac{R_k}K\big)^{1/2}\big)\Big]+\frac{R_{k+1}}{2KC}\\
&\le C_1\Big[\rho^{-N/2}+\big(\frac K{R_k}\big)^{-c_0}\Big]+\frac{R_{k+1}}{2KC}.
\end{align*}
From this, one can easily adapt the proof to reach the conclusion of the bridging lemma with $\alpha(2,N)=N/2$ (note that except for the first and last values of $k$, the right-hand side
 is bounded by $2C_1\rho^{-N/2}$). After this, Proposition~\ref{prop:not bridged} follows in the same way. Also, note that the assumption that $p_\lambda(K,\theta)\le 1/C$ is harmless as otherwise the statement of Proposition~\ref{prop:not bridged}
 is obvious. Once Proposition~\ref{prop:not bridged} has been obtained, the rest of the proof of the theorem is the same as for Theorem~\ref{thm:FK}.
\end{proof}

We conclude our paper by the proof of Proposition~\ref{prop:technical}.

\begin{figure}[!htp]
\begin{center}
\includegraphics[width=1.00\textwidth]{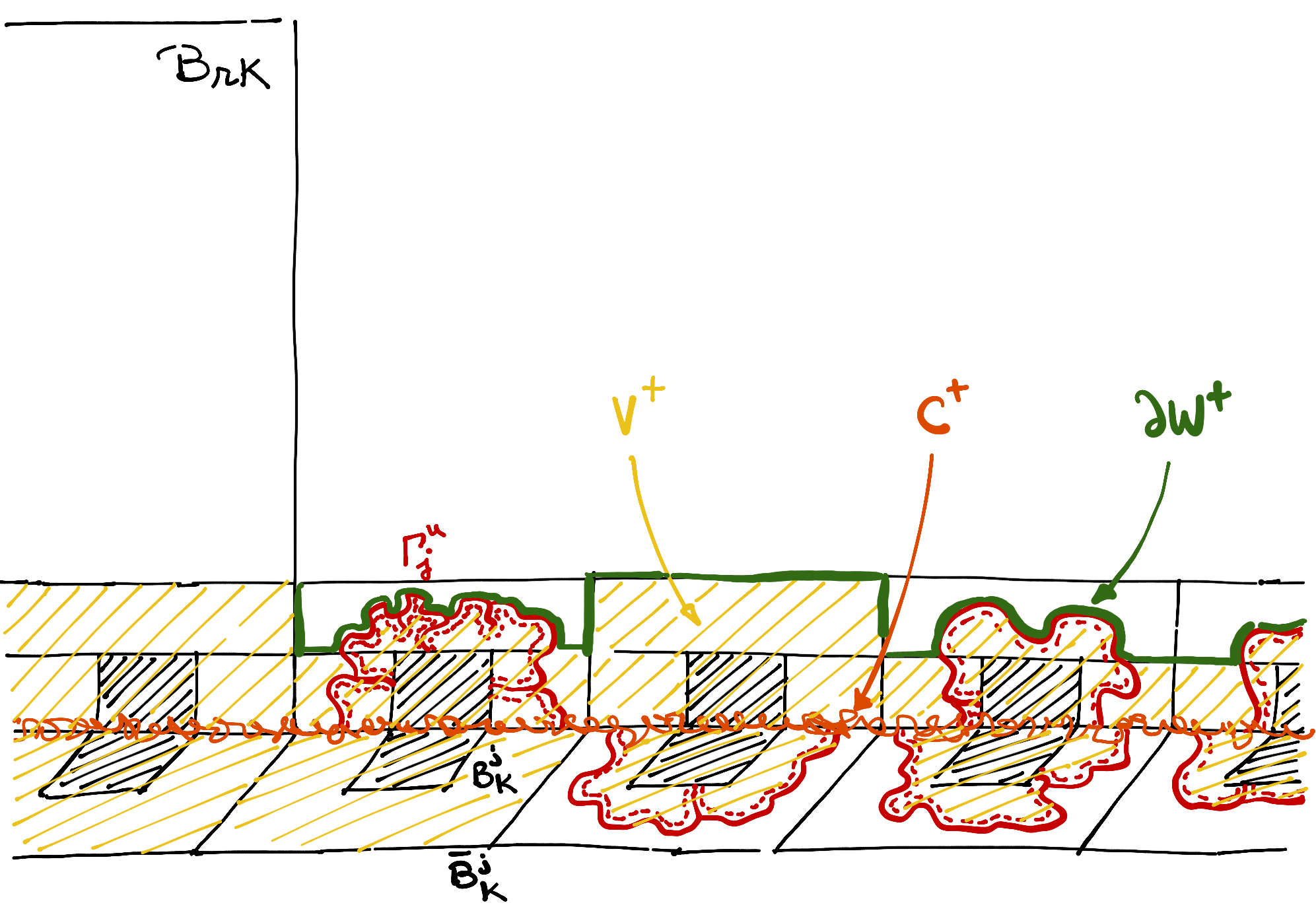}
\end{center}
\caption{A picture of the path $\Gamma^u_j$ as well as $\mathbf C^+$, $\mathbf V^+$ and $\mathbf W^+=\mathbf V^+\cap(B_{RK}\setminus B_{rK})$ respectively in red, yellow and green. }\label{fig:Ising}
\end{figure}

\begin{proof}[Proof of Proposition~\ref{prop:technical}]
Set $s:=\sqrt r$. We may assume $p_\lambda(K,\theta)< \tfrac1{2R}$ as otherwise the statement is obvious. 
Let $E$ be the event that all the $K$-blocks $B_K^i$ are $\theta$-good for $|i|< R$. By definition of $p_\lambda(K,\theta)$, 
\begin{equation}\label{eq:aa1}
\mathbb P_{B_{KR}}^0[E]\ge 1-2Rp_\lambda(K,\theta).
\end{equation}
Below, we recommend to take a look at Fig.~\ref{fig:Ising}. Introduce $\overline B_K^j$ to be the set of vertices at a $\ell^\infty$ distance at most $K$ from $B_K^j$. Define $\calB^{u\pm}$ to be the sets of indexes $j\in[-R,R]$ divisible by 6, positive or negative depending on whether $\pm$ is $+$ or $-$, for which $B_K^j$ is surrounded in $\overline B_K^j\cap \mathbb H^u$ by a circuit in $\omega$ connecting $\mathbf{C}(B_K^{j+2})$ to $\mathbf C(B_K^{j-2})$. Call the inner-most such circuit $\Gamma_j^u=\Gamma_j^u(\omega)$. Let 
\[
E':=\bigcap_{\substack{|i|\le R\\ u=1,\dots ,N \\
a\in\{\pm\}}}\big\{|\calB^{ua}\cap[is,(i+1)s]|\ge c_0 s\big\}
.
\]
Adapting the anchoring lemma and using that $p_\lambda(K,\theta)\le \tfrac12$, we deduce that for some $c_1,c_2>0$ small enough,
\[
\mathbb P_{\overline B_K^j}^0[j\in \calB^{ua},\forall u=1,\dots,N]\ge c_1(1-p_\lambda(K,\theta))^{2N}\ge c_2.
\]
Using the comparison between boundary conditions, we may compare to independent random variables to get that $|\calB^{ua}\cap[is,(i+1)s]|$ dominates a binomial random variable with parameters $s$ and $c_1$, so that for some constant $c_0>0$ small enough and independent of everything else,
\begin{equation}\label{eq:aa2}
\mathbb P_{B_{RK}}^0[E']\ge1- 2R\exp[-c_0s].
\end{equation}
We deduce from \eqref{eq:aa1} and \eqref{eq:aa2} that it suffices to show that 
\begin{equation}
\mathbb P_{B_{KR}}^0[F(Kr,KR)^c|E\cap E']\le C_3\Big[\big(\frac rR\big)^{N/2}+r^{-c_3}\Big].
\end{equation}
We now focus on deriving this inequality.

On $E\cap E'$, call $\mathbf C^+$ the union of the $\mathbf C(B_K^j)$ for $6\le j\le R$, which since $\theta>\tfrac34$ is made of one single cluster. From now on, let $\mathbf V^+=\mathbf V^+(\omega)$ be the (random) subset of $B_{RK}$ obtained as 
\begin{itemize}
\item the union of all the $B_K^j$ for $6\le j\le R$;
\item the vertices of $\mathbb H^u$ surrounded by the $\Gamma_j^u(\omega)$ for every $6\le j\le R$ divisible by 6 for which the path $\Gamma_j^u(\omega)$ exists for every $u$;
\item the union of the $\overline B_K^j$ for the remaining $6\le j\le R$ which are divisible by 6.
\end{itemize}
Similarly, one defines $\mathbf C^-$ and $\mathbf V^-$ with $-R\le j\le -6$ instead of $6\le j\le R$. Also, set $\mathbf E(\omega)$ be the set of edges with both endpoints in $\mathbf V^+\cup \mathbf V^-$.

Condition on $\omega_{|\mathbf E(\omega)}=\xi$ for some configuration $\xi \in\{0,1\}^{\mathbf E(\omega)}$ belonging to $E\cap E'$ (by this we mean that any configuration coinciding with $\xi$ on $E(\omega)$ is in $E\cap E'$).  
Let $\Omega^u$ be the graph induced by the  edges in $(B_{RK}\cap\mathbb H^u)\setminus (B_{K}\cup \mathbf V^+\cup \mathbf V^-)$  and $\xi^u$ be the boundary condition on $\Omega^u$ obtained from  the configuration equal to $\xi$ on $\mathbf E(\omega)$, and 0 on the remaining part of $B_{RK}\setminus \Omega^u$. 

The comparison between boundary conditions implies that for every $\xi\in E\cap E'$,
\[
\mathbb P_{B_{KR}}^0[F(Kr,KR)^c|\omega_{|\mathbf E(\omega)}=\xi]\le\prod_{u=1}^N(1-\mathbb P_{\Omega^u}^{\xi^u}[\mathbf C^-\longleftrightarrow \mathbf C^+]).
\]
It therefore suffices to prove that each term on the right is smaller than $C_4(r/R)^{1/2}$. From now on, we call a pair $(\Omega,\psi)$, with $\Omega$ a subset of $\mathbb H$ and $\psi$ a boundary condition on $\Omega$ {\em possible} if there exists $\xi\in E\cap E'$ and $u$ such that $\Omega=\Omega^u$ and $\psi=\xi^u$.  In this case we write $\mathbf V^\pm$ for the corresponding set (they can be read off from $\Omega$ and $\psi$ in a unique fashion). 

Let $\mathbf W^\pm$ be the intersection of $\mathbf V^\pm$ with $\Lambda_{RK}\setminus\Lambda_{rK}$. Consider the boundary condition $\xi\cup1$ obtained from $\xi$ by wiring all the vertices in $\mathbf W^+$ together, and all those in $\mathbf W^-$ together.
For every possible $(\Omega,\xi)$, going to the complement implies that
\begin{align*}
\mathbb P_{\Omega}^{\xi\cup1}[\mathbf W^+\longleftrightarrow \mathbf W^-]\ge 1-a_{wired}^+(rK,RK)\ge 1-C(r/R)^{1/2}.
\end{align*}
The following lemma will therefore conclude the proof.\end{proof}

\begin{lemma}\label{lem:estimate one arm}
There exist $c,C\in(0,\infty)$ independent of everything such that for every possible $(\Omega,\xi)$,
\[
\mathbb P_{\Omega}^{\xi}[\mathbf C^- \longleftrightarrow \mathbf C^+]\ge \mathbb P_{\Omega}^{\xi\cup 1}[\mathbf W^+\longleftrightarrow \mathbf W^-](1-Cr^{-c}).
\]
\end{lemma}

As in the case of Theorems \ref{th.BookFactor} and \ref{th.Hbernoulli} for which we sketched the proofs earlier, 
the derivation of this lemma will rely on the {\em random-current representation}. 
A first key property will be the identity~\eqref{eq:crucial RC} which follows from the {\em switching lemma}. 

We will use also a second property of our model, which is a coupling between sourceless random-current and FK percolation with parameter $q=2$. More precisely, consider the coupling $\phi$ obtained by considering $\n\sim \mathbf P^A_{G,\beta}$ and $\omega$ obtained from $\n$ by setting
\[
\omega_e=\sup\{\mathbbm 1[\n>0],\eta_e\},
\]
where $(\eta_e:e\in E)$ is an independent family of Bernoulli random variables of parameter $1-e^{-\beta}$. Then, one has that under $\phi$, $\omega\sim\mathbb P_{G,p,2}^0[\cdot|\calF_A]$ with $p:=1-e^{-2\beta}$.
 While the recipe to get $\omega$ from $\n$ is obvious, let us mention that in the other direction, for $A=\emptyset$, one may recover the edges on which $\n$ is odd by taking a uniform even subgraph of $\omega$ (see \cite{ADTW19,GriJan09,LupWer15}).

We are now in a position to prove the lemma.

\begin{proof}
Consider the graphs $G$ obtained from $\Omega$ by identifying all the vertices in $\mathbf C^\pm$ into two vertices $g^\pm$, and $\overline G$  obtained from $G$ by identifying vertices in $\mathbf W^\pm\setminus\mathbf C^\pm$ to $g^\pm$. Note that $G$ can be seen as a subgraph of $\overline G$ where the latter is obtained from the former by adding edges with infinite coupling constants (or equivalently infinitely many edges with standard coupling constant) between the vertices of $\mathbf W^\pm\setminus \mathbf C^\pm$ and $g^\pm$). 

We have that 
\begin{align*}
\mathbb P_{\Omega}^{\xi}[\mathbf C^+\longleftrightarrow \mathbf C^-]&=\mu_{G,\beta_c}[\sigma_{g^+}\sigma_{g^-}],\\
\mathbb P_{\Omega}^{1}[\mathbf W^+\longleftrightarrow \mathbf W^-]&=\mu_{\overline G,\beta_c}[\sigma_{g^+}\sigma_{g^-}],
\end{align*}
so that \eqref{eq:crucial RC} gives that 
\[
\frac{\mu_{G,\beta_c}[\sigma_{g^+}\sigma_{g^-}]}{\mu_{\overline G,\beta_c}[\sigma_{g^+}\sigma_{g^-}]}=\mathbf P_{\overline G}^{\{g^+,g^-\}}\otimes\mathbf P_G^{\emptyset}[\exists\text{ path of }\n_1+\n_2>0\text{ in }G\text{ from $g^-$ to $g^+$}]
\]
and it therefore suffices to bound the probability on the right-hand side. 

First, observe that the coupling between random-current and FK percolation implies that
\begin{align*}
&\mathbf P_{\overline G}^{\{g^+,g^-\}}[\exists j\in[r,r+r^{3/4}]:B_K^j\text{ connected to distance $rK$ in }\mathbf 1[\n>0]\setminus g^+]\\
&\le 
\mathbb P^0_{\overline G,p_c,2}[\exists j\in[r,r+r^{3/4}]:B_K^j\text{ connected to distance $rK$ in }G\setminus g^+|g^-\leftrightarrow g^+]\\
&\le \frac{a_{\C}^+(r^{3/4}K,rK)}{\mathbb P^0_{\overline G,p_c,2}[g^-\leftrightarrow g^+]}\le Cr^{-c}.
\end{align*}
Similarly for $R-r^{3/4}\le j\le R$, $-R\le j\le -R+r^{3/4}$, and $-r-r^{3/4}\le j\le -r$.
We therefore may restrict to realizations of $\mathbf n_1$ that necessarily contain a path $\gamma$ of $\n_1(e)>0$ from $g^+$ to $g^-$, going say from $\overline B_K^i$ to $\overline B_K^j$, with 
\[
-R+r^{3/4}\le i\le -r-r^{3/4}\qquad\text{and}\qquad r+r^{3/4}\le j\le R-r^{3/4}.\]
As a consequence, it suffices to prove that in $\mathbf n_2$, with large probability there exists a path of $\n_2(e)>0$ from vertices in $\mathbf C^+$ respectively on the left and right of $\overline B_K^j$ (call the two parts $\mathcal L$ and $\mathcal R$). The same estimate will also holds for $-R\le i\le -r$.

In order to do that, we write $\mathbf n$ instead of $\mathbf n_2$ and use the increasing coupling $\phi$ between $\mathbf n$ and the random-cluster model $\omega$ described before the proof. It is sufficient to prove that 
\begin{equation}\label{eq:h1}
\mathbb P_\Omega^\xi[\text{there exists $c_0\log r$ disjoint clusters going from $\calL$ to $\calR$}]\ge1-1/r^{c_0}.\end{equation}
Indeed, on this event, one may divide clusters in pairs, and observe that each pair of clusters contains a loop (with one path in one cluster and the other in the other one) of $\mathbf n>0$ connecting $\calL$ to $\calR$ with probability at least $1/2$ thanks to the fact that the odd part of $\mathbf n$ is obtained from $\omega$ by taking an even subgraph of $\omega$ uniformly at random. Therefore, the probability that there exists no path at all will be smaller than $1/r^{c_0}+2^{-(c_0/2)\log r}$.

To prove \eqref{eq:h1}, first shift the whole configuration by $(-Kj,0)$ in order to recenter everything around 0. 
On the one hand, crossing estimates imply that for $k$ such that $s=\sqrt r\le 2^k\le r^{3/4}$,
\begin{equation}\label{eq:ab1}
\mathbb P_\Omega^1[\Lambda_{2^kK}\longleftrightarrow \partial\Lambda_{2^{k+1}K}]\le 1-c
\end{equation}
for some constant $c$ independent of everything. On the other hand, if $\Omega_k$ denotes the intersection of $\Omega$ with the annulus $\Lambda_{2^{k+1}K}\setminus\Lambda_{2^kK}$, we want to prove that 
\begin{equation}\label{eq:ab2}
\mathbb P_{\Omega_k}^{0}[\mathcal L\longleftrightarrow \mathcal R]\ge c.
\end{equation}
 This will conclude the proof by observing that \eqref{eq:ab1} and \eqref{eq:ab2} together with the spatial Markov property and the comparison between boundary conditions easily imply  \eqref{eq:h1} by following a proof quite similar to the bridging lemma.

To prove \eqref{eq:ab2}, we use a second-moment method very similar to the proof of the anchoring lemma. Let $\mathbf N$ be the number of pairs $-\tfrac532^k\le a\le -\tfrac432^k$ and $\tfrac432^k\le b\le \tfrac532^{k+1}$ with $\Gamma_a$ connected to $\Gamma_b$ (recall the definition of these paths from the previous section, and remember that the whole configuration has been shifted by $(Kj,0)$). Note that by definition of a possible pair $(\Omega,\xi)$ (since $\xi$ belonged to $E'$), there are of order $c_0(2^k)^2$ pairs of $(a,b)$. Also, an easy comparison between boundary conditions and use of crossing estimates implies that 
\begin{align*}\label{}
\mathbb P_{\Omega_k}^{0}[\mathbf M] & \ge c_0(2^k)^2\min_{a,b}\mathbb P_{\Omega_k}^{0}[\Gamma_a\longleftrightarrow \Gamma_b]\ge c_1(2^k)^2a_{wired}^+(K,2^kK)^2.
\end{align*}
In the other direction, the comparison between boundary conditions and a standard use of the quasi-mulitiplicativity property gives that
\begin{align*}\label{}
\mathbb P_{\Omega_k}^{0}[\mathbf M^2] & \le \sum_{a,a',b,b'}\mathbb P_\mathbb H^0[\Gamma_a\longleftrightarrow\Gamma_b,\overline\Gamma_{a'}\longleftrightarrow\overline\Gamma_{b'}]\le C_0(2^k)^4a_{wired}^+(K,2^kK)^4.
\end{align*}
\end{proof}

\subsection{Proof of  Theorem~\ref{thm:Ising2}.}
$ $

The proof of the decoupling between the pages of $\Book_3$ stated in Theorem \ref{thm:Ising2} follows easily by combining the proof of Proposition \ref{prop:technical} together with the (sketch) of proof of the decoupling property from Theorem \ref{th.BookFactor}. Let us shortly explain why we have an error term $1-O_{m}((\log L)^{-c})$ in Theorem~\ref{thm:Ising2} versus $1-O_{m}(L^{-c})$ in Theorem \ref{th.BookFactor}. To prove Theorem \ref{thm:Ising2}, we rely on  the multiscale framework used throughout the paper. In particular, if all points $\{x_1,\ldots,x_m\}$ in $A \subset \Book_3$ are at a distance at least $L$ from $\Z$, consider $n$ such that 
\[
K_n \leq L < K_{n+1}\,.
\]
Recall (footnote below~\eqref{eq:13}) that $K_n = (n!)^3 C_1^n$. This implies in particular that $n \geq (\log L)^{1/2}$, when $L$ is large enough. Let us now proceed as in the proof of Theorem \ref{th.BookFactor} and let $u$ be the furthest point on the left of the joint line $\Z$ of an $\n_1$ cluster emanating from, say the first point $x_1\in A$ (other possible points being connected to $\Z$ via $\n_1$ are handled similarly  by union bound). We now consider the blocks $B_{K_{n-1}^i}$ at scales $n-1$ around the point $u\in \Z$. From our inductive proof, we know that each of these $(n-1)$-block is good with probability at least $1-u_{n-1} \geq 1 - \frac{1}{1000 C_{n-1}^2}$. (See the estimates on $u_n$ below~\eqref{eq:13}). This implies that with probability at least $1-\tfrac1{100 C_n}$, all $(n-1)$-blocks $B_{K_{n-1}}^i$ around the point $u$ and up to distance $K_n=C_n K_{n-1} \leq L$  are good. We can now use this overlapping chain of good blocks as in the proof of Lemma \ref{lem:estimate one arm} to produce a bridging with the random current $\n_1$ with probability at least $1-O((C_n)^{-c})$ which is the same as $1-O((\log L)^{-\tilde c})$ and thus concludes our proof. \qed

\begin{remark}\label{}
Note that by going further into smaller scales $K_{n-\ell} \ll K_n \leq L < K_{n+1}$ and by replacing the power-law control $\{u_n\}_{n\geq 1}$ below~\eqref{eq:13} by an exponentially decaying control in $n$, one may obtain if needed better correction terms in Theorem~\ref{thm:Ising2}.
\end{remark}

\bibliographystyle{alpha}

\end{document}